\documentclass[12pt, draftclsnofoot, onecolumn]{IEEEtran}

\usepackage{amssymb}
\usepackage{amsmath}
\usepackage{cite}
\usepackage{url}
\usepackage{xcolor}
\usepackage{cite,graphicx,amsmath,amssymb}
\usepackage{subfigure}
\usepackage{fancyhdr}
\usepackage{mdwmath}
\usepackage{mdwtab}
\usepackage{caption}
\usepackage{amsthm}
\usepackage{setspace}
\usepackage{hyperref}
\usepackage{algorithm}
\usepackage{algorithmic}
\usepackage{multirow}
\usepackage{makecell}
\hypersetup{colorlinks=false}

\graphicspath{{./src/img/}}

\newtheorem{remark}{Remark}
\newtheorem{theorem}{Theorem}

\newtheorem{lemma}{Lemma}

\newtheorem{corollary}{Corollary}

\newtheorem{proposition}{Proposition}

\allowdisplaybreaks

\title{Exploiting Sensing Signal in ISAC: \\ A NOMA Inspired Scheme}
\author{

        Zhaolin~Wang,~\IEEEmembership{Graduate Student Member,~IEEE,}
        Xidong~Mu,~\IEEEmembership{Member,~IEEE,}
        Yuanwei~Liu,~\IEEEmembership{Senior Member,~IEEE,}
        and Zhiguo~Ding,~\IEEEmembership{Fellow,~IEEE}

\thanks{Zhaolin Wang, Xidong Mu, and Yuanwei Liu are with the School of Electronic Engineering and Computer Science, Queen Mary University of London, London E1 4NS, U.K. (e-mail: zhaolin.wang@qmul.ac.uk, xidong.mu@qmul.ac.uk, yuanwei.liu@qmul.ac.uk).}
\thanks{Zhiguo Ding is with the School of Electrical and Electronic Engineering, The University of Manchester, Manchester M13 9PL, UK (e-mail: zhiguo.ding@manchester.ac.uk).}
\vspace{-1.5cm}
}

\begin{document}

\maketitle

\begin{abstract}
    A non-orthogonal multiple access (NOMA)-inspired integrated sensing and communication (ISAC) framework is proposed, where a dual-functional base station (BS) transmits the composite communication and sensing signals. In contrast to treating the sensing signal as a harmful interference to communication, in this work, multiple beams of the sensing signal are exploited to convey extra information streams based on the concept of NOMA. Then, each communication user detects the extra information streams and the existing legacy information streams with the aid of successive interference cancellation (SIC). Based on the proposed framework, a multiple-objective optimization problem (MOOP) is formulated for designing the transmit beamforming subject to the total transmit power constraint, which characterizes the trade-off between the communication throughput and sensing beampattern accuracy. For the general multiple-user scenario, the formulated MOOP is firstly converted to a single-objective optimization problem (SOOP) via the $\epsilon$-constraint method. Then, a double-layer block coordinate descent (BCD) algorithm is proposed by employing fractional programming (FP) and successive convex approximation (SCA) to find a high-quality sub-optimal solution. For the special single-user scenario, the globally optimal solution can be obtained by transforming the MOOP into a convex quadratic semidefinite program (QSDP). Moreover, it is rigorously proved that 1) in the multiple-user scenario, the proposed NOMA-inspired ISAC framework always outperforms the state-of-the-art sensing-interference-cancellation (SenIC) ISAC frameworks by further exploiting sensing signals for delivering extra information streams; 2) in the special single-user scenario, the proposed NOMA-inspired ISAC framework achieves the same performance as the existing SenIC ISAC frameworks, which reveals that the coordination of sensing interference is not necessarily required in this case. Numerical results verify the theoretical results and show that exploiting one beam of the sensing signal for delivering an extra information stream is sufficient for the proposed NOMA-inspired ISAC framework.
\end{abstract}


\section{Introduction}
The beyond fifth generation (B5G) and sixth generation (6G) wireless networks are envisaged to enable a host of emerging applications like unmanned aerial vehicles (UAVs), autonomous driving, virtual and augmented reality, Internet of vehicles, and smart city \cite{letaief2019roadmap}. Therefore, in contrast to the past five generations of wireless networks that mainly support wireless communication, the B5G and 6G require a paradigm shift to the integration of multiple functions including communications, sensing, control, and computing \cite{saad2019vision}. Based on this vision, the concept of integrated sensing and communication (ISAC)\footnote{Note that the “sensing” in ISAC refers to wireless radio sensing.} has emerged and attracted growing attention in both academia \cite{liu2021integrated, cui2021integrating} and industries \cite{wild2021joint, tan2021integrated}. The goal of ISAC is to integrate communication and sensing via the same platform and the same spectrum, which is capable of increasing resource efficiency and achieving mutual benefits. On the one hand, the resources for communication and sensing are shared in the ISAC, thus improving the utilization efficiency in terms of spectrum, hardware, energy, and costs \cite{liu2021integrated}. On the other hand, the two functions can be interplayed with each other for realizing win-win operations, such as the sensing-assisted pilot-free communication channel estimation \cite{cui2021integrating} and the communication-assisted high-precision sensing \cite{liu2021integrated}.

\subsection{Prior Works}
Due to the resource sharing between the two functions of communication and sensing, there is a natural performance trade-off in ISAC. Therefore, the bound of the performance trade-off between communication and sensing was investigated in \cite{chiriyath2017radar} from the information-theoretic perspective, where the radar estimation rate was developed based on the rate-distortion theory to unify the metrics of communication and sensing. The authors of \cite{guerci2015joint} introduced the concept of radar capacity based on the Hartley capacity measure and investigated the total communication-radar capacity of the integrated system. Waveform design is another research aspect and plays a key role for ISAC to achieve harmonious resource sharing and integration gain \cite{liu2021integrated}. Early works of waveform design in ISAC focused on the single-antenna systems. For instance, the authors of \cite{moghaddasi2016multifunctional} designed a joint waveform for supporting the communication and sensing in a time-division manner, where the trapezoidal frequency-modulated continuous waveform was exploited for sensing in the radar cycle while the modulated single frequency carrier appeared in the communication cycle. In addition, the potential of the linear frequency modulated (LFM) waveform and orthogonal frequency division modulated (OFDM) waveform for ISAC was investigated in \cite{saddik2007ultra} and \cite{sturm2011waveform}, respectively.

Given the limitations of the single-antenna technique, the multiple-antenna technique, which has been respectively successfully employed in both communication \cite{tse2005fundamentals} and sensing \cite{li2007mimo}, is adopted in the recent research contributions on ISAC. By exploiting spatial degrees of freedom (DoFs) provided by the multiple-antenna technique, different beams can be generated to carry out multiple-user communication and multiple-target sensing. For example, the authors of \cite{liu2018mu} proposed two strategies for facilitating ISAC, namely the separated deployment and the shared deployment, where the transmit beamforming in each deployment was designed to achieve the high-quality sensing beampattern while satisfying the minimum communication requirement. As further research of the shared deployment, the closed-form optimal transmission beamforming for minimizing the inter-user interference subject to different sensing criteria was derived in \cite{liu2018toward}, based on which the communication and sensing performance trade-off was investigated. Moreover, the authors of \cite{chen2021joint} introduced a general Pareto optimization framework and proposed a bisection algorithm to find the optimal trade-off point of communication and sensing in terms of signal-to-interference-and-noise ratio (SINR) and the difference between peak and sidelobe (DPSL), respectively. Most recently, the application of non-orthogonal multiple access (NOMA) in the ISAC system was investigated in \cite{wang2022noma}, where multiple communication users are supported via superposition coding (SC) and successive interference cancellation (SIC). The authors of \cite{mu2021noma} considered a more sophisticated application scenario of NOMA in ISAC, namely sensing targets require the multicast messages from the BS, which can be non-orthogonally superimposed and transmitted simultaneously with the unicast messages for communication users.

However, only communication signal was exploited in \cite{liu2018mu, liu2018toward, chen2021joint, wang2022noma, mu2021noma} for facilitating ISAC. Despite having the advantage that no additional interference is received at the communication users, it may lead to sensing beampattern distortion due to the lack of transmitting DoFs, especially when the number of communication users is less than the number of transmit antennas \cite{liu2020beamforming}. As a remedy, the authors of \cite{liu2020beamforming} proposed to employ extra dedicated sensing signals to provide full transmit DoFs for sensing. The authors of \cite{liu2021cram} also employed the dedicated sensing signals to generate the full-rank transmit covariance matrix, which is exploited to estimate the parameters of the extended target via the maximum likelihood estimation. In \cite{hua2021optimal}, the optimal transmit beamforming design is studied, where it is proved and demonstrated that canceling the interference from sensing signal to communication can significantly improve the ISAC performance. 

\begin{table*}[t!]
    \caption{Our contributions compared with the state-of-the-art}
    \centering
    \setlength{\arrayrulewidth}{1pt}
    \begin{tabular}{|l|c|c|c|c|}
        \hline
                                              & \cite{liu2018mu, liu2018toward, chen2021joint, wang2022noma, mu2021noma}   & \cite{liu2020beamforming,liu2021cram}   &\cite{hua2021optimal}   & \textbf{Our work}                        \\ \hline
        Communication signal                  & $\surd$                                                                      & $\surd$                                 & $\surd$                & $\surd$                        \\ \hline
        Sensing signal                        & $\times$                                                                    & $\surd$                                 & $\surd$                & $\surd$                        \\ \hline
        Sensing interference cancellation     & $\times$                                                                    & $\times$                                & $\surd$                & $\surd$                        \\ \hline
        Sensing signal for communication      & $\times$                                                                    & $\times$                                & $\times$               & $\surd$                        \\ \hline
    \end{tabular}
    \label{table:contribution}
\end{table*}

\subsection{Motivation and Contributions}
Compared to the system which merely employs the communication signals for achieving the two functions, it has been shown that additional dedicated sensing signals are generally essential to better facilitate sensing in ISAC \cite{liu2020beamforming, liu2021cram, hua2021optimal}. This, however, also causes significant challenges in how to coordinate the sensing signals for supporting communications. In the existing ISAC frameworks, the sensing signal is regarded as a harmful interference to communication, which is mitigated through the beamforming design \cite{liu2020beamforming, liu2021cram} or the ideal interference-cancellation receiver \cite{hua2021optimal}. To the best of the authors' knowledge, the potential of exploiting the sensing signal to further convey information for benefiting the communication has not been investigated. This provides the main motivation for this work.

Note that the superposition of the communication and sensing signals in ISAC is a kind of \textit{non-orthogonal} resource sharing, which shares a similar idea with NOMA \cite{liu2017non, liu2018multiple}. For NOMA, the information-bearing signals are superimposed at the transmitter and the co-channel interference is partially or totally eliminated by the SIC at the receiver \cite{liu2021evolution}. Although there have been some works on the interplay between NOMA and ISAC \cite{wang2022noma} and \cite{mu2021noma}, they mainly focused on exploiting NOMA to achieve multiple access to different kinds of communication users. However, in this work, we go beyond the scope of multiple access and aim to use the principle of NOMA to coordinate the communication signal and the sensing signal. Based on this idea, we propose a NOMA-inspired ISAC framework. The key principle is that the sensing signal is further exploited for extra communication transmission and added to the existing legacy communication signal based on the mechanism of NOMA. The main benefits of the proposed NOMA-inspired ISAC framework are summarized as follows. On the one hand, since the sensing signal is information-bearing, the SIC technique can be employed by communication users to remove the interference suffering from the sensing signal. On the other hand, the sensing signal is not only used for sensing but also used for delivering extra information streams, thus also benefiting communications. Based on this framework, we investigate the corresponding transmit beamforming designs and the communication-sensing trade-off. The main contributions of this work are summarized below, which are boldly and explicitly contrasted to the state-of-the-art in Table \ref{table:contribution}.

\begin{itemize}
    \item We propose a novel NOMA-inspired ISAC framework, where a dual-functional base station (BS) transmits the information-bearing sensing signal and the legacy communication signals based on the concept of NOMA. The sensing signal is exploited for both extra information transmission and target sensing, and its interference at the communication users is eliminated via SIC. To investigate the fundamental trade-off between communication and sensing, we formulate a multiple-objective optimization problem (MOOP) of the communication throughput and the sensing beampattern matching error subject to the total transmit power constraint.

    \item For the multiple-user scenario, we convert the formulated MOOP to a single-objective optimization problem (SOOP) via the $\epsilon$-constraint method. To solve the resultant non-convex SOOP, we propose a double-layer block coordinate descent (BCD) algorithm employing fractional programming (FP) and successive convex approximation (SCA) to obtain a suboptimal solution.
    
    \item For the special single-user scenario, we reformulate the communication throughput and convert the related MOOP to a quadratic semidefinite program (QSDP) by employing 
    the $\epsilon$-constraint method, which is shown to be convex and can be globally optimally solved. 

    \item We theoretically prove that for achieving the same sensing performance, the proposed framework can enhance the communication performance compared to the existing sensing-interference-cancellation (SenIC) ISAC frameworks in the multiple-user scenario. Moreover, for the single-user scenario, we rigorously prove that the performance achieved by all schemes becomes the same, indicating that sensing interference coordination is not needed in this case.

    \item Our numerical results verify that in comparison to the existing ISAC framework, the proposed NOMA-inspired ISAC framework achieves a larger trade-off region and better transmit beampattern for sensing in the multiple-user scenario while having the same performance in the single-user scenario. Furthermore, it also shows that in the proposed framework, exploiting only one beam of the sensing signal for the extra information transmission is sufficient to achieve the maximum communication throughput.
\end{itemize}

\subsection{Organization and Notation}

The rest of this paper is organized as follows. In Section \ref{sec:system_model}, the proposed NOMA-inspired ISAC framework is presented. Then, a MOOP of transmit beamforming between communication throughput and sensing beampattern accuracy is formulated. In Section \ref{sec:solution}, a suboptimal solution to the formulated problem in the general multiple-user scenario is obtained by the proposed double-layer FP-SCA-based BCD algorithm, while a globally optimal solution in the special single-user scenario is obtained via QSDP. Then, the properties of the proposed framework are studied for the two scenarios in comparison to the state-of-the-art ISAC frameworks. In Section \ref{sec:results}, the numerical results are presented to characterize the proposed NOMA-inspired ISAC framework. Finally, this paper is concluded in Section \ref{sec:conclusion}. 

\textit{Notations:} 
Scalars, vectors, and matrices are denoted by the lower-case, bold-face lower-case,
and bold-face upper-case letters, respectively; 
$\mathbb{C}^{N \times M}$ and $\mathbb{R}^{N \times M}$ denotes the space of $N \times M$ complex and real matrices, respectively;
$a^*$ and $|a|$ denote the conjugate and magnitude of scalar $a$;  
$\mathbf{a}^H$ denotes the conjugate transpose of vector $\mathbf{a}$; 
$\mathbf{A} \succeq 0$ means that matrix $\mathbf{A}$ is positive semidefinite; 
$\mathrm{rank}(\mathbf{A})$, $\|\mathbf{A}\|_*$, and $\|\mathbf{A}\|_2$ denote the rank, nuclear norm, and spectral norm of matrix $\mathbf{A}$, respectively;
$\mathbb{E}[\cdot]$ denotes the statistical expectation; 
$\mathcal{CN}(\mu, \sigma^2)$ denotes the distribution of a circularly symmetric complex Gaussian (CSCG) random variable with mean $\mu$ and variance $\sigma^2$.

\section{System Model and Problem Formulation} \label{sec:system_model}

\subsection{System Model}
\begin{figure} [ht!]
  \centering
  \includegraphics[width=0.7\textwidth]{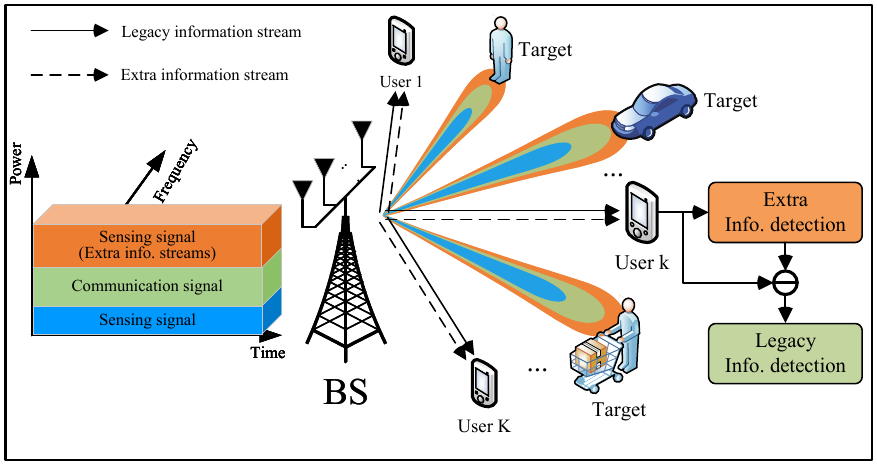}
   \caption{Illustration of the proposed NOMA-inspired ISAC framework.}
   \label{fig:system_model}
   \vspace{-0.3cm}
\end{figure}

As shown in Fig. \ref{fig:system_model}, a NOMA-inspired ISAC framework is proposed, which consists of a dual-functional $N$-antenna BS, $K$ single-antenna communication users indexed by $\mathcal{K} = \{1,\dots,K\}$, and multiple sensing targets. According to \cite{liu2020beamforming}, to realize the full DoFs of target sensing, the BS transmits the superimposed communication signals and sensing signals, which is given by 
\begin{equation} \label{eqn:transmis_signal}
  \mathbf{x}(t) = \mathbf{W}_c \mathbf{c}(t) + \mathbf{s}(t) = \sum_{k \in \mathcal{K}} \mathbf{w}_{c,k} c_k(t) + \mathbf{s}(t).
\end{equation}
Here, $\mathbf{W}_c = [\mathbf{w}_{c,1},\dots,\mathbf{w}_{c,K}] \in \mathbb{C}^{N \times K}$ denotes the performers for transmitting the information stream $\mathbf{c}(t) = [c_1(t),...,c_K(t)]^T \in \mathbb{C}^{K \times 1}$ at time index $t$, and $\mathbf{s}(t) \in \mathbb{C}^{N \times 1}$ denotes the sensing signal. We assume the Gaussian signaling for the information stream with zero mean and unit power, namely $\mathbb{E}[ \mathbf{c}(t) \mathbf{c}(t)^H ] = \mathbf{I}_K$ and $\mathbb{E}[ \mathbf{c}(t) \mathbf{s}(t)^H ] = \mathbf{0}_{K \times N}$. The covariance matrix of sensing signal $\mathbf{s}(t)$ is denoted by $\mathbf{R}_{\mathbf{s}} = \mathbb{E}[\mathbf{s}(t)\mathbf{s}(t)^H] \succeq 0$. Multiple-beam transmission is employed to construct the sensing signal. In this case, the sensing signal can be decomposed into multiple sensing beams via the eigenvalue decomposition:
\begin{equation}
    \mathbf{R}_{\mathbf{s}} = \sum_{i=1}^{N} \lambda_i \mathbf{v}_i \mathbf{v}_i^H  = \sum_{i=1}^{N} \mathbf{w}_{s,i} \mathbf{w}_{s,i}^H,
\end{equation}
where $\lambda_i \in \mathbb{R}$ is the eigenvalue and $\mathbf{v}_i \in \mathbb{C}^{N \times 1}$ is the corresponding eigenvector. The vector $\mathbf{w}_{s,i} = \sqrt{\lambda_i} \mathbf{v}_i$ is the transmit beamformer for the $i$-th beam of the sensing signal. Without loss of generality, we assume $\lambda_1 \ge \lambda_2 \ge \dots \ge \lambda_{N}$. Inspired by NOMA, we embed extra information stream $\tilde{\mathbf{c}}(t) = [\tilde{c}_1(t),\dots,\tilde{c}_M(t)]^T \in \mathbb{C}^{M \times 1}$ into the sensing beam corresponding to the largest $M$ eigenvalues on top of the existing legacy communication signals, as shown in Fig. \ref{fig:NOMA-inspired ISAC}. These extra information streams can be used for supporting additional services, such as enhancing the service quality of the users with high priority or broadcasting common messages. In this case, the sensing signal can be rewritten as follows:
\begin{equation} \label{eqn:sensing_multicast}
    \mathbf{s} = \mathbf{W}_s \tilde{\mathbf{c}}(t) + \tilde{\mathbf{s}}(t) = \sum_{i \in \mathcal{M}} \mathbf{w}_{s,i} \tilde{c}_i  + \tilde{\mathbf{s}}(t),
\end{equation}
where $\mathcal{M}=\{1,\dots,M\}$, $\mathbf{W}_s = [\mathbf{w}_{s,1},\dots,\mathbf{w}_{s,M}]$, and $\tilde{\mathbf{s}}$ denotes the remaining sensing signal without information. We also assume the Gaussian signaling with zero mean and unit power for the extra information streams $\tilde{\mathbf{c}}(t)$, namely $\mathbb{E}[ \tilde{\mathbf{c}}(t) \tilde{\mathbf{c}}(t)^H ] = \mathbf{I}_M$ and $\mathbb{E}[ \tilde{\mathbf{c}}(t) \tilde{\mathbf{s}}(t)^H ] = \mathbf{0}_{M \times N}$. As a consequence, the covariance matrix $\mathbf{R}_{\tilde{\mathbf{s}}}$ of $\tilde{\mathbf{s}}(t)$ is given as follows :
\begin{align} \label{eqn:remaining_interference}
    \mathbf{R}_{\tilde{\mathbf{s}}} &= \mathbb{E}\Big[ \big(\mathbf{s}(t) - \mathbf{W}_s \tilde{\mathbf{c}}(t)\big) \big(\mathbf{s}(t) - \mathbf{W}_s \tilde{\mathbf{c}}(t)\big)^H \Big] = \sum_{i=Q+1}^{N} \mathbf{w}_{s,i} \mathbf{w}_{s,i}^H.
\end{align}

\subsubsection{Communication Model}

\begin{figure} [t!]
  \centering
  \includegraphics[width=0.6\textwidth]{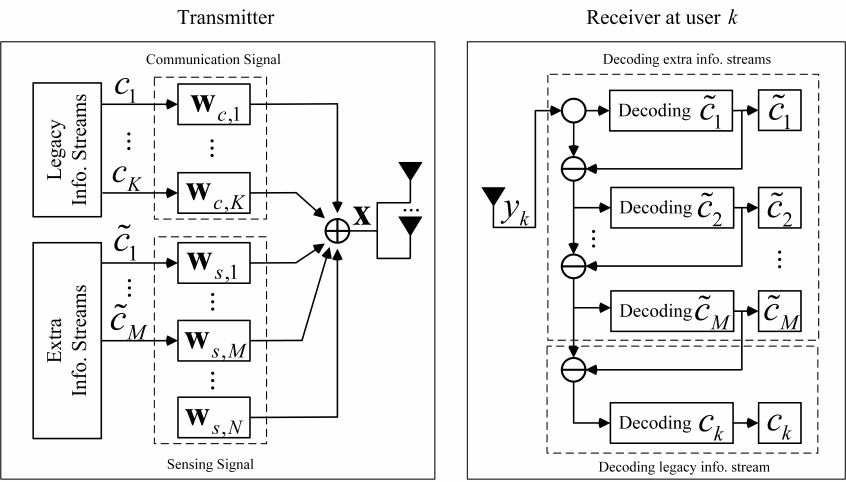}
   \caption{Block diagram of the transmitter and the receiver at user $k$ in the NOMA-inspired ISAC framework.}
   \label{fig:NOMA-inspired ISAC}
   \vspace{-0.5cm}
\end{figure}

Given the transmit signal in \eqref{eqn:transmis_signal} and \eqref{eqn:sensing_multicast}, the received signal at user $k$ is given by
\begin{align} \label{eqn:received}
  y_k(t) = & \mathbf{h}_k^H \mathbf{W}_c \mathbf{c}(t) + \mathbf{h}_k^H \mathbf{W}_s \tilde{\mathbf{c}}(t) + \mathbf{h}_k^H \tilde{\mathbf{s}}(t) + n_k(t) \nonumber \\
  = &\mathbf{h}_k^H \mathbf{w}_{c,k} c_k(t) + \underbrace{\sum_{j \in \mathcal{K}\setminus k}\mathbf{h}_k^H \mathbf{w}_{c,j} c_j(t)}_{\text{inter-user interference}} + \underbrace{\sum_{i \in \mathcal{M}}\mathbf{h}_k^H \mathbf{w}_{s,i} \tilde{c}_i(t) + \mathbf{h}_k^H \tilde{\mathbf{s}}(t)}_{\text{sensing signal}} + n_k(t),
\end{align} 
where $\mathbf{h}_k \in \mathbb{C}^{N \times 1}$ denotes the BS-user channel and $n_k(t) \sim \mathcal{CN}(0, \sigma_n^2)$ denotes the CSCG noise with variance $\sigma_n^2$. The decoding procedure at user $k$ is shown in Fig. \ref{fig:NOMA-inspired ISAC}. In particular, the information-bearing sensing signals $\mathbf{w}_{s,i} \tilde{c}_i(t), \forall i \in \mathcal{M}$ are successively decoded via SIC firstly. Thus, the achievable rate of the $i$-th stream $\tilde{c}_i(t)$ embedded into the sensing signal at user $k$ is given by
\begin{equation} \label{eqn:multicast_rate_1}
  \tilde{R}_{i,k}^{N} = \log_2 \left( 1 + \frac{|\mathbf{h}_k^H \mathbf{w}_{s,i}|^2}{ \sum_{q \in \mathcal{M}, q > i}|\mathbf{h}_k^H \mathbf{w}_{s,q}|^2 + \sum_{j \in \mathcal{K}} |\mathbf{h}_k^H \mathbf{w}_{c,j}|^2 + \mathbf{h}_k^H \mathbf{R}_{\tilde{\mathbf{s}}} \mathbf{h}_k + \sigma_n^2} \right).
\end{equation} 
To successfully eliminated the information-bearing sensing signals at all the $K$ communication users by SIC, the effective achievable rate of the $i$-th streams $\tilde{c}_i(t)$ needs to be limited by the user achieving the lowest rate, which is given by
\begin{equation} \label{eqn:multicast_rate_2}
    \tilde{R}_i^{N} = \min \{R_{i,1}^{N}, R_{i,2}^{N}, \dots, R_{i,K}^{N} \}.
\end{equation} 
After decoding and removing all streams $\tilde{c}_i(t), \forall i \in \mathcal{M}$ via SIC, there is no interference from information-bearing sensing signals to the existing legacy communication stream $\mathbf{w}_{c,k}c_k(t)$ for user $k$. Thus, the achievable rate of the stream $c_k(t)$ at user $k$ is given by
\begin{equation} \label{eqn:unicast_rate}
    R_k^{N} = \log_2 \left( 1 + \frac{| \mathbf{h}_k^H \mathbf{w}_{c,k}|^2}{ \sum_{j \in \mathcal{K} \setminus k} |\mathbf{h}_k^H \mathbf{w}_j|^2 +  \mathbf{h}_k^H \mathbf{R}_{\tilde{\mathbf{s}}} \mathbf{h}_k + \sigma_n^2} \right),
\end{equation}
The communication throughput of the proposed NOMA-inspired ISAC system is the sum rate of the extra information streams delivered by the sensing signal and the existing legacy information streams, which is given by 
\begin{equation} \label{eqn:throughput_I}
    R^{N} = \sum_{i \in \mathcal{M}} \tilde{R}_i^{N} + \sum_{k \in \mathcal{K}} R_k^{N}.
\end{equation}

\subsubsection{Sensing Model}
In the ISAC system, both communication and sensing signals can be exploited for target sensing, since the BS has complete knowledge of the transmit signals. In this case, there is no communication-to-sensing interference. Assuming planar wave and line-of-sight (LoS) propagation, the signal in the direction of $\theta$ can be modeled as 
\begin{equation} \label{eqn:y_s}
  y_s(\theta, t) = \mathbf{a}^H(\theta) \mathbf{x},
\end{equation} 
$\mathbf{a}(\theta_l) = [1, e^{j\frac{2\pi}{\lambda}d\sin({\theta_l})},...,e^{j\frac{2\pi}{\lambda}d(N-1)\sin({\theta_l})}]^T$ is the steering vector of direction $\theta_l$ when the uniform linear array (ULA) is equipped at the BS. Then, the echo signal reflected by the sensing target in the direction of $\theta$ is given by 
\begin{equation}
  \mathbf{y}_s(\theta, t) = \beta \mathbf{b}(\theta) \mathbf{a}^H(\theta) \mathbf{x} +  \mathbf{n}_s,
\end{equation}
where $\beta$ denotes the complex reflecting coefficient of the target plus the round-trip path loss, $\mathbf{b}(\theta)$ denotes the steering vector of the receiver at the BS, and $\mathbf{n}_s$ denotes the complex Gaussian noise. According to \cite{stoica2007probing}, the transmit signal needs to be carefully designed for the desired sensing performance. Typically, the transmit signal is optimized to approximate the desired sensing beampattern. Based on \eqref{eqn:y_s}, the transmit beampattern achieved by the signal $\mathbf{x}$ in the direction of $\theta$ is given by 
\begin{equation} \label{eqn:transmit_beampattern}
  P(\theta, \mathbf{R}_\mathbf{x}) = \mathbb{E}[|y_s(\theta, t)|^2] = \mathbf{a}^H(\theta) \mathbf{R}_{\mathbf{x}} \mathbf{a}(\theta),
\end{equation} 
where $\mathbf{R}_{\mathbf{x}}$ denotes transmit covariance matrix:
\begin{align} \label{eqn:transmit_covariance}
  \mathbf{R}_{\mathbf{x}} &=\mathbb{E}[\mathbf{x}\mathbf{x}^H]= \sum_{k \in \mathcal{K}} \mathbf{w}_{c,k} \mathbf{w}_{c,k}^H + \mathbf{R}_{\mathbf{s}} \nonumber\\
  &= \sum_{k \in \mathcal{K}} \mathbf{w}_{c,k} \mathbf{w}_{c,k}^H + \sum_{i \in \mathcal{M}} \mathbf{w}_{s,i} \mathbf{w}_{s,i}^H + \mathbf{R}_{\tilde{\mathbf{s}}}.
\end{align}

In practice, the desired sensing beampattern is designed according to the sensing requirements \cite{stoica2007probing}. For example, if the sensing system
has no information about the target and works in the detecting mode, an isotropic beampattern is desired, i.e., the power is uniformly distributed among all directions.
However, when the sensing system has prior information of targets and works in the tracking mode, the beampattern is expected to have the dominant peak power in the target directions.
Let $\{\phi(\theta_l)\}_{l=1}^L$ denotes a pre-designed sensing beampattern over an angular grid $\{\theta_l\}_{l=1}^L$ covering the detector's angular range $[-\frac{\pi}{2}, \frac{\pi}{2}]$. Then, the sensing performance can be evaluated by the following loss function \cite{stoica2007probing}:
\begin{equation}
  L(\delta, \mathbf{R}_\mathbf{x}) = \frac{1}{L} \sum_{l=1}^{L}\left| \delta \phi(\theta_l) - \mathbf{a}^H(\theta_l) \mathbf{R}_{{\bf x}} \mathbf{a}(\theta_l)  \right|^2,
\end{equation}
where $\delta$ is a scaling factor. It is worth noting that minimizing $L(\delta, \mathbf{R}_\mathbf{x})$ is in conflict with maximizing the communication throughput $R^N$, especially when sensing works in the tracking mode. The reason can be explained as follows. Maximizing the communication throughput requires high transmit power in the user-desired directions, e.g., the directions of scatters and reflectors in the BS-user channels. However, minimizing $L(\delta, \mathbf{R}_\mathbf{x})$ will force the BS to transmit high power in the target directions and low power in the non-target directions. Generally, due to the randomness of the BS-user channel, the user-desired directions have a high probability to fall into the non-target directions. In this case, the tailored transmit beamforming design is required to facilitate both communication and sensing simultaneously.

\begin{remark}
  \emph{The main benefits of the proposed NOMA-inspired ISAC framework can be summarized as follows. Firstly, by embedding the information into part of the sensing signal, the sensing interference on the existing legacy communication signals can be eliminated via SIC, see \eqref{eqn:received} and \eqref{eqn:unicast_rate}.
  Secondly, the sensing signal is further exploited for communication by delivering extra information streams, see \eqref{eqn:sensing_multicast} and \eqref{eqn:throughput_I}. Thirdly, the sensing signal provides more DoFs to achieve high quality sensing beampattern, see \eqref{eqn:transmit_beampattern} and \eqref{eqn:transmit_covariance}. Last but not least, the proposed NOMA-inspired ISAC framework facilities the \emph{double benefits} of the sensing signal, namely enhancing both communication and sensing.}
\end{remark}

\subsection{Problem Formulation}
Given the proposed NOMA-inspired ISAC framework, it is shown in the previous subsection that there are two conflicting objectives, namely maximizing the communication throughput and minimizing the beampattern matching error. For maximizing the communication throughput, the related single-objective optimization problem (SOOP) subject to the total transmit power is formulated as follows:

\emph{Problem 1 (Throughput Maximization)}:
\begin{subequations}
  \begin{align}
    \text{(P1):} \quad \max_{\mathbf{W}_c, \mathbf{W}_s, \mathbf{R}_{\tilde{\mathbf{s}}} } \quad &\sum_{i \in \mathcal{M}} \tilde{R}_i^{N} + \sum_{k \in \mathcal{K}} R_k^{N} \\
    \mathrm{s.t.} \quad & \mathrm{tr}(\mathbf{W}_c \mathbf{W}_c^H + \mathbf{W}_s \mathbf{W}_s^H + \mathbf{R}_{\tilde{\mathbf{s}}}) \le P_t,\\
    & \mathbf{R}_{\tilde{\mathbf{s}}} \succeq 0.
  \end{align}
\end{subequations}
where $P_t$ denotes the transmit power budget of the BS.
Unlike the communication functionality, the sensing functionality requires the equality power constraint for achieving optimal sensing performance. Thus, the SOOP for the beampattern matching error minimization can be formulated as follows:

\emph{Problem 2 (Beampattern Matching Error Minimization)}:
\begin{subequations}
  \begin{align}
    \text{(P2):} \quad \min_{\delta, \mathbf{W}_c, \mathbf{W}_s, \mathbf{R}_{\tilde{\mathbf{s}}} } \quad & L(\delta, \mathbf{R}_\mathbf{x}) = \frac{1}{L} \sum_{l=1}^{L}\left| \delta \phi(\theta_l) - \mathbf{a}^H(\theta_l) \mathbf{R}_{{\bf x}} \mathbf{a}(\theta_l)  \right|^2 \\
    \mathrm{s.t.} \quad & \mathrm{tr}(\mathbf{W}_c \mathbf{W}_c^H + \mathbf{W}_s \mathbf{W}_s^H + \mathbf{R}_{\tilde{\mathbf{s}}}) = P_t,\\
    & \mathbf{R}_{\tilde{\mathbf{s}}} \succeq 0.
  \end{align}
\end{subequations}
In the problem (P2), if the power constraint is not forced to achieve equality, we can always obtain the optimal solution $\delta=0$ and $\mathbf{R}_{\mathbf{x}} = \mathbf{0}$, such that $L(\delta, \mathbf{R}_\mathbf{x})$ is minimized. Apparently, it is the undesired solution that needs to be avoided. Then, based on (P1) and (P2), a multiple-objective optimization problem (MOOP) is formulated for the purpose of investigating the trade-off between the two conflicting objectives, which is given by 

\emph{Problem 3 (Multiple Objective Optimization)}:
\begin{subequations}
  \begin{align}
    \text{(P3):} \quad \text{Q1}: \max_{\mathbf{W}_c, \mathbf{W}_s, \mathbf{R}_{\tilde{\mathbf{s}}} } \quad &\sum_{i \in \mathcal{M}} \tilde{R}_i^{N} + \sum_{k \in \mathcal{K}} R_k^{N} \\
    \text{Q2}: \min_{\delta, \mathbf{W}_c, \mathbf{W}_s, \mathbf{R}_{\tilde{\mathbf{s}}} } \quad & \frac{1}{L} \sum_{l=1}^{L}\left| \delta \phi(\theta_l) - \mathbf{a}^H(\theta_l) \mathbf{R}_{{\bf x}} \mathbf{a}(\theta_l)  \right|^2 \\
    \label{c:transmit_power}
    \mathrm{s.t.} \quad & \mathrm{tr}(\mathbf{W}_c \mathbf{W}_c^H + \mathbf{W}_s \mathbf{W}_s^H + \mathbf{R}_{\tilde{\mathbf{s}}}) = P_t,\\
    \label{c:semidefinite_sensing_signal}
    & \mathbf{R}_{\tilde{\mathbf{s}}} \succeq 0.
  \end{align}
\end{subequations}

\section{Proposed Solutions and Property Analysis} \label{sec:solution}
In this section, we first propose a double-layer FP-SCA-based BCD algorithm to solve the problem (P3) of two conflicting objectives in the general multiple-user scenario. To be specific, we first transform the problem (P3) to a SOOP by adopting the $\epsilon$-constraint method \cite{miettinen2012nonlinear}, where Q1 is treated as the main objective while Q2 is transformed into the constraint set. Then, the resultant SOOP is decomposed into three disjoint blocks via the FP methods \cite{shen2018fractional1, shen2018fractional2} and solved by exploiting BCD and SCA \cite{sun2016majorization} to obtain a suboptimal solution. Moreover, for the spacial single-user scenario, we obtain a globally optimal solution by treating Q2 as the main objective and transforming Q1 into the constraint. The resultant problem is a convex quadratic semidefinite programming (QSDP) problem and can be optimally solved. Finally, the property of the proposed framework is analyzed by comparing it with the state-of-the-art ISAC frameworks.

\subsection{Suboptimal Solution to (P3) for the General Multiple-User Scenario} \label{sec:multiple}
To address the conflicting objectives in (P3), 
the $\epsilon$-constraint method is exploited, based on which (P3) is transformed into a SOOP as follows: 
\begin{subequations}
    \begin{align}
      \label{c:obj_1}
      \text{(P3.1):} &\max_{\delta, \mathbf{W}_c, \mathbf{W}_s, \mathbf{R}_{\tilde{\mathbf{s}}} } \quad \sum_{i \in \mathcal{M}} \tilde{R}_i^{N} + \sum_{k \in \mathcal{K}} R_k^{N} \\
      \mathrm{s.t.} \quad & \frac{1}{L} \sum_{l=1}^{L}\left| \delta \phi(\theta_l) - \mathbf{a}^H(\theta_l) \mathbf{R}_{{\bf x}} \mathbf{a}(\theta_l)  \right|^2 \le \epsilon_1, \\ 
      &\eqref{c:transmit_power}, \eqref{c:semidefinite_sensing_signal}.
    \end{align}
\end{subequations}
It is worth noting that the whole Pareto frontier of the two objectives can be obtained by varying the value of $\epsilon_1$ and solving the corresponding optimization problem \cite{miettinen2012nonlinear}.
However, it is general challenging to find the globally optimal solution to (P3.1) since the optimization variables $\mathbf{W}_c$, $\mathbf{W}_s$, and $\mathbf{R}_{\tilde{\mathbf{s}}}$ are highly coupled in the non-convex objective function \eqref{c:obj_1} and quadratic equality constraint \eqref{c:transmit_power}.
As a remedy, we propose a double-layer FP-SCA-based BCD algorithm to obtain a suboptimal solution. To start with, we introduce the auxiliary variables $\mathbf{W}_{c,k} = \mathbf{w}_{c,k} \mathbf{w}_{c,k}^H, \forall k \in \mathcal{K}$, where $\mathbf{W}_{c,k} \succeq 0$ and $\mathrm{rank}(\mathbf{W}_{c,k})=1$, and $\mathbf{W}_{s,i} = \mathbf{w}_{s,i} \mathbf{w}_{s,i}^H, \forall i \in \mathcal{M}$, where $\mathbf{W}_{s,i} \succeq 0$ and $\mathrm{rank}(\mathbf{W}_{s,i})=1$. 
Then, the transmit covariance matrix in \eqref{eqn:transmit_covariance} becomes
\begin{equation}
  \mathbf{R}_\mathbf{x} = \sum_{k \in \mathcal{K}} \mathbf{W}_{c,k} + \sum_{i \in \mathcal{M}} \mathbf{W}_{s,i} + \mathbf{R}_{\tilde{\mathbf{s}}}.
\end{equation}
Moreover, the achievable rate of the extra information streams and the existing communication stream at user $k$ can be rewritten as
\begin{align}
  & \tilde{R}_{i,k}^{N} = \log_2 \Big( 1 + \frac{\mathbf{h}_k^H \mathbf{W}_{s,i} \mathbf{h}_k }{ \sum_{q \in \mathcal{M}, q > i} \mathbf{h}_k^H \mathbf{W}_{s,q} \mathbf{h}_k + \sum_{j \in \mathcal{K}} \mathbf{h}_k^H \mathbf{W}_{c,j} \mathbf{h}_k + \mathbf{h}_k^H \mathbf{R}_{\tilde{\mathbf{s}}} \mathbf{h}_k + \sigma_n^2 } \Big),\\
  & R_k^{N} = \log_2 \Big( 1 + \frac{\mathbf{h}_k^H \mathbf{W}_{c,k} \mathbf{h}_k }{ \sum_{j \in \mathcal{K} \setminus k} \mathbf{h}_k^H \mathbf{W}_{c,j} \mathbf{h}_k + \mathbf{h}_k^H \mathbf{R}_{\tilde{\mathbf{s}}} \mathbf{h}_k + \sigma_n^2} \Big).
\end{align}
Then, in order to remove the $\min\{\cdot\}$ function in the extra rate $\tilde{R}_i^{N}$, we transform (P3.1) to an equivalent problem, which is given by 
\begin{subequations}
    \begin{align}
      \text{(P3.2):} & \max_{ \mathbf{\Theta} } \quad f_0(\mathbf{\Theta}) = \sum_{i \in \mathcal{M}} \tilde{r}_i + \sum_{k \in \mathcal{K}} R_k^{N} \\
      \label{c:min_rate}
      \mathrm{s.t.} \quad & \tilde{R}_{i,k}^{N} \ge \tilde{r}_i \ge 0, \forall i \in \mathcal{M}, \forall k \in \mathcal{K},\\
      \label{c:e_constraint}
      & \frac{1}{L} \sum_{l=1}^{L}\left| \delta \phi(\theta_l) - \mathbf{a}^H(\theta_l) \mathbf{R}_{{\bf x}} \mathbf{a}(\theta_l)  \right|^2 \le \epsilon_1,\\
      & \sum_{k \in \mathcal{K}} \mathrm{tr}(\mathbf{W}_{c,k}) + \sum_{i \in \mathcal{M}} \mathrm{tr}(\mathbf{W}_{s,i}) + \mathrm{tr}(\mathbf{R}_{\tilde{\mathbf{s}}}) = P_t,\\
      \label{c:rank_1}
      &\mathbf{W}_{s,i} \succeq 0, \mathrm{rank}(\mathbf{W}_{s,i}) = 1, \forall i \in \mathcal{M},\\
      \label{c:rank_1_2}
      &\mathbf{W}_{c,k} \succeq 0, \mathrm{rank}(\mathbf{W}_{c,k}) = 1, \forall k \in \mathcal{K},\\
      \label{c:semidefinite}
      &\mathbf{R}_{\tilde{\mathbf{s}}} \succeq 0,
    \end{align}
\end{subequations}
where $\tilde{r}_i$ is the auxiliary variable and $\mathbf{\Theta} = (\delta, \{\tilde{r}_i\}, \{\mathbf{W}_{c,k}\}, \{\mathbf{W}_{s,i}\}, \mathbf{R}_{\tilde{\mathbf{s}}})$ denotes all the optimization variables. Next, we reformulated problem (P3.2) by invoking the FP approach proposed in \cite{shen2018fractional1} and \cite{shen2018fractional2}. 
Based on FP, (P3.2) can be decomposed into three disjoint blocks and solved iteratively by BCD. Specifically, in BCD, two blocks can be updated with the closed-form expressions and the remaining one can be updated via SCA \cite{sun2016majorization}.

\subsubsection{FP approach for problem reformulation} \label{sec:FP}
Let $\gamma_k$ denote the SINR term in $R_k^{N}$, i.e., $R_k^{N} = \log_2 (1 + \gamma_k)$.
Then, by applying the Lagrangian dual transform proposed in \cite{shen2018fractional2} and introducing an auxiliary variable $\boldsymbol{\alpha} = [\alpha_1,\dots,\alpha_K]^T \in \mathbb{R}^{K \times 1}$,
$R_k^{N}$ can be transformed to 
\begin{equation}
  r_k^N = \log_2(1 + \alpha_k) - \alpha_k + \frac{(1+\alpha_k) \gamma_k}{1 + \gamma_k}.
\end{equation} 
By replacing $R_k^{N}$ with $r_k^N$ in the problem (P3.2), we obtain a new optimization problem, which is given by 
\begin{subequations} 
  \begin{align}
    \text{(P3.3):} \max_{\boldsymbol{\alpha}, \mathbf{\Theta}} \quad &
    f_1(\boldsymbol{\alpha}, \mathbf{\Theta}) = \sum_{i \in \mathcal{M}} \tilde{r}_i + \sum_{k \in \mathcal{K}} r_k^N \\
    \mathrm{s.t.} \quad & \eqref{c:min_rate} - \eqref{c:rank_1}.
  \end{align}
\end{subequations}

\begin{lemma} \label{lemma_1}
  \emph{
    (P3.2) and (P3.3) are equivalent. The optimal $\boldsymbol{\alpha}$ is given by
    \begin{equation} \label{eqn:opt_alpha}
      \alpha_k^\star = \gamma_k, \forall k \in \mathcal{K}.
    \end{equation}
  }
\end{lemma}

\begin{proof}
  Please refer to \cite[Theorem 3]{shen2018fractional2}.
\end{proof}

Lemma \ref{lemma_1} shows the equivalence between (P3.2) and (P3.3). However, the objective function of (P3.3) is still non-convex due to 
the sum-of-ratio term, which is denoted by 
\begin{align}
  g( \boldsymbol{\alpha}, \mathbf{\Theta}) &= \sum_{k \in \mathcal{K}} \frac{(1+\alpha_k) \gamma_k}{1 + \gamma_k} = \sum_{k \in \mathcal{K}} \frac{(1+\alpha_k) \mathbf{h}_k^H \mathbf{W}_{c,k} \mathbf{h}_k }{ \sum_{j \in \mathcal{K}} \mathbf{h}_k^H \mathbf{W}_{c,j} \mathbf{h}_k + \mathbf{h}_k^H \mathbf{R}_{\tilde{\mathbf{s}}} \mathbf{h}_k + \sigma_n^2 }.
\end{align}
Then, the objective function of (P3.3) can be rewritten as 
\begin{align}
  f_1(\boldsymbol{\alpha}, \mathbf{\Theta}) = \sum_{i \in \mathcal{M}} \tilde{r}_i + v(\boldsymbol{\alpha}) + g(\boldsymbol{\alpha}, \mathbf{\Theta}),
\end{align}
where $v(\boldsymbol{\alpha}) = \sum_{k \in \mathcal{K}}(\log_2(1 + \alpha_k) - \alpha_k))$. In order to solve the non-convexity of $g( \boldsymbol{\alpha}, \mathbf{\Theta})$, 
we transform it into a new function via the quadratic transform proposed in \cite{shen2018fractional1}, which is given by 

\begin{align}
  h(\boldsymbol{\alpha}, \boldsymbol{\beta}, \mathbf{\Theta}) = &\sum_{k \in \mathcal{K}} \Big( 2 \mathrm{Re} \big\{ \beta^* \sqrt{(1+\alpha_k)\mathbf{h}_k^H \mathbf{W}_{c,k} \mathbf{h}_k } \big\} \Big) \nonumber \\
  &- \sum_{k \in \mathcal{K}} |\beta_k|^2 \Big( \sum_{j \in \mathcal{K}} \mathbf{h}_k^H \mathbf{W}_{c,j} \mathbf{h}_k + \mathbf{h}_k^H \mathbf{R}_{\tilde{\mathbf{s}}} \mathbf{h}_k + \sigma_n^2 \Big),
\end{align} 
where $\boldsymbol{\beta} = [\beta_1,\dots,\beta_K]^T \in \mathbb{C}^{K \times 1}$ is the auxiliary variable.
Then, a new optimization problem can be obtained from (P3.3) by replacing $g(\boldsymbol{\alpha}, \mathbf{\Theta})$ with $h(\boldsymbol{\alpha}, \boldsymbol{\beta}, \mathbf{\Theta})$:
\begin{subequations} 
  \begin{align}
    \text{(P3.4):} \max_{\boldsymbol{\alpha}, \boldsymbol{\beta}, \mathbf{\Theta}} \quad &
    f_2(\boldsymbol{\alpha}, \boldsymbol{\beta}, \mathbf{\Theta}) \!=\! \sum_{i \in \mathcal{M}} \tilde{r}_i + v(\boldsymbol{\alpha}) + h(\boldsymbol{\alpha}, \boldsymbol{\beta}, \mathbf{\Theta}) \\
    \mathrm{s.t.} \quad & \eqref{c:min_rate} - \eqref{c:rank_1}.
  \end{align}
\end{subequations}

\begin{lemma} \label{lemma_2}
  \emph{
    (P3.3) and (P3.4) are equivalent. The optimal $\boldsymbol{\beta}$ is given by 
  \begin{equation} \label{eqn:opt_beta}
    \beta_k^\star = \frac{\sqrt{(1+\alpha_k)\mathbf{h}_k^H \mathbf{W}_{c,k} \mathbf{h}_k}}{\sum_{j \in \mathcal{K}} \mathbf{h}_k^H \mathbf{W}_{c,j} \mathbf{h}_k + \mathbf{h}_k^H \mathbf{R}_{\tilde{\mathbf{s}}} \mathbf{h}_k + \sigma_n^2}.
  \end{equation}  
  }
\end{lemma}

\begin{proof}
  Please refer to \cite[Theorem 2]{shen2018fractional1}.
\end{proof}
According to Lemma \ref{lemma_1} and Lemma \ref{lemma_2}, (P3.2) is equivalent to (P3.4). It can be observed that the optimization variables of 
the problem (P3.4) can be decomposed into three disjoint blocks $\boldsymbol{\alpha}$, $\boldsymbol{\beta}$, and $\mathbf{\Theta}$.   
Therefore, it can be solved by exploiting BCD through the following procedure:
\begin{enumerate}
  \item Initialize $n=0$ and a feasible $\mathbf{\Theta}^n$;
  \item Update $\boldsymbol{\alpha}^{n+1}$ by \eqref{eqn:opt_alpha} with $\mathbf{\Theta}^n$;
  \item Update $\boldsymbol{\beta}^{n+1}$ by \eqref{eqn:opt_beta} with $\boldsymbol{\alpha}^{n+1}$ and $\mathbf{\Theta}^n$;
  \item Update $\mathbf{\Theta}^{n+1}$ by solving (P3.4) with $\boldsymbol{\alpha}^{n+1}$ and $\boldsymbol{\beta}^{n+1}$.
\end{enumerate}
Here, $\boldsymbol{\alpha}^n$, $\boldsymbol{\beta}^n$, and $\mathbf{\Theta}^n$ denote the obtained optimization variables in the $n$-th iterations of BCD.
However, despite that the objective function of (P3.4) is concave with respect to $\mathbf{\Theta}$ for the fixed $\boldsymbol{\alpha}$ and $\boldsymbol{\beta}$,
the constraints \eqref{c:min_rate}, \eqref{c:rank_1} and \eqref{c:rank_1_2} are non-convex, which can be solved by invoking SCA \cite{sun2016majorization}.

\subsubsection{SCA for updating $\mathbf{\Theta}$} \label{sec:SCA}
To solve the non-convexity of the constraint \eqref{c:min_rate}, we firstly rewrite $\tilde{R}_{i,k}$ as follows:
\begin{align}
  \tilde{R}_{i,k}
  = \log_2 \left( \mathbf{h}_k^H \left(\mathbf{A}_i + \mathbf{W}_{s,i} \right) \mathbf{h}_k + \sigma_n^2 \right) \underbrace{ - \log_2 \left( \mathbf{h}_k^H \mathbf{A}_i \mathbf{h}_k + \sigma_n^2 \right)}_{t_{i,k}},
\end{align}
where $\mathbf{A}_i = \sum_{q \in \mathcal{M}, q > i}\mathbf{W}_{s,q}  + \sum_{j \in \mathcal{K}}  \mathbf{W}_{c,j} +  \mathbf{R}_{\tilde{\mathbf{s}}}$.
It can be observed that the non-convexity lies in the term $t_{i,k}$. Let $\mathbf{A}_i^n$ denote the value of $\mathbf{A}_i$ in the $n$-th iteration of BCD. Then, we adopt a surrogate function $\tilde{t}_{q,k}$ constructed using the first-order Taylor expansion of $t_{i,k}$  at 
point $\mathbf{A}_i^n$, which is given by 
\begin{align}
  \tilde{t}_{q,k} \! \triangleq \! - \log_2 \left( \mathbf{h}_k^H \mathbf{A}_i^n \mathbf{h}_k \!+\! \sigma_n^2 \right) \!-\! \frac{\mathbf{h}_k^H \left( \mathbf{A}_i - \mathbf{A}_i^n \right) \mathbf{h}_k}{\left( \mathbf{h}_k^H \mathbf{A}_i^n \mathbf{h}_k + \sigma_n^2 \right) \ln 2},
\end{align}
As $t_k$ is a convex function of the optimization variables $\mathbf{A}_i$, the inequality $t_{i,k} \ge \tilde{t}_{q,k}$ holds.
Then, $\tilde{R}_{i,k}$ is bounded by 
\begin{equation}
  \tilde{R}_{i,k} \ge \overline{R}_{i,k}^N \triangleq  \log_2 \left( \mathbf{h}_k^H \!\left( \mathbf{A}_i + \mathbf{W}_{s,i} \right) \mathbf{h}_k + \sigma_n^2 \right) + \tilde{t}_{q,k}.
\end{equation} 
Therefore, the constraint \eqref{c:min_rate} can be approximated by $\overline{R}_{i,k}^N \ge \tilde{r}_i \ge 0, \forall i \in \mathcal{M}, \forall k \in \mathcal{K}$, which is convex. 

To solve the non-convex rank-one constraints \eqref{c:rank_1} and \eqref{c:rank_1_2}, the semidefinite relaxation (SDR) \cite{luo2010semidefinite} is generally applied by omitting the rank-one constraint and transforming the original problem to a convex semidefinite program (SDP). In this case, the global optimum of the resultant SDR optimization problem can be efficiently obtained by the convex solvers like CVX \cite{cvx}. Nevertheless, as the rank-one property of the matrices $\mathbf{W}_{s,i}, \forall i \in \mathcal{M}$ and $\mathbf{W}_{c,k}, \forall k \in \mathcal{K}$ is not gauranteed in SDR, the rank-one solutions can be reconstructed via the eigenvalue decomposition or the Gaussian randomization. However, these methods do not ensure the feasibility of the reconstructed rank-one solution and may lead to performance loss. Therefore, we propose to solve the rank-one constraint by a penalty-based scheme \cite{mu2021noma}, where the rank-one constraints are transformed into a penalty term in the objective function. Firstly, the rank-one constraint can be equivalently transformed into
\begin{subequations}
  \begin{align}
    &\| \mathbf{W}_{s,i} \|_* - \| \mathbf{W}_{s,i} \|_2 = 0, \forall i \in \mathcal{M},\\
    &\| \mathbf{W}_{c,k} \|_* - \| \mathbf{W}_{c,k} \|_2 = 0, \forall k \in \mathcal{K},
  \end{align}
\end{subequations}
where $\| \cdot \|_*$ and $\| \cdot \|_2$ denote the nuclear norm and spectral norm, respectively. 
The nuclear norm is the sum of all the eigenvalues, while the spectral norm equals the largest eigenvalue. 
Thus, for the positive semidefinite matrices $\mathbf{W}_{s,i}$ and $\mathbf{W}_{c,k}$, 
it holds that $\| \mathbf{W}_{s,i} \|_* - \| \mathbf{W}_{s,i} \|_2 \ge 0$ and $\| \mathbf{W}_{c,k} \|_* - \| \mathbf{W}_{c,k} \|_2 \ge 0$.   
Then, the nearly rank-one matrices can be obtained by minimizing the difference between the nuclear norm and spectral norm.
Based on this observation, we introduce a penalty term to the objective function of the problem (P3.4), which is given by 
\begin{align}
  \varUpsilon (\mathbf{W}_{s,i}, \mathbf{W}_{c,k}) = \sum_{i \in \mathcal{M}} \big(\| \mathbf{W}_{s,i} \|_* - \|\mathbf{W}_{s,i}\|_2 \big) + \sum_{k \in \mathcal{K}} \big(\| \mathbf{W}_{c,k} \|_* - \|\mathbf{W}_{c,k}\|_2 \big).
\end{align}
Then, the new optimization problem with the penalty term for updating $\mathbf{\Theta}$ with $\boldsymbol{\alpha}^{n+1}$ and $\boldsymbol{\beta}^{n+1}$ is given as follows:
\begin{subequations} 
  \begin{align}
    \text{(P3.5):} \max_{\mathbf{\Theta}} \quad &
    f_2(\boldsymbol{\alpha}^{n+1}, \boldsymbol{\beta}^{n+1}, \mathbf{\Theta}) - \frac{1}{\zeta} \varUpsilon(\mathbf{W}_{s,i}, \mathbf{W}_{c,k})\\
    \label{c:multi_rate_sca}
    \mathrm{s.t.} \quad & \overline{R}_{i,k}^N \ge \tilde{r}_i \ge 0, \forall k \in \mathcal{K},\\
    & \eqref{c:e_constraint} - \eqref{c:semidefinite},
  \end{align}
\end{subequations}
where $\zeta$ is the regularization parameter. However, problem (P3.5) is still non-convex due to the terms  $-\| \mathbf{W}_{s,i} \|_2$ and $-\| \mathbf{W}_{c,k} \|_2$ in the penalty term $\varUpsilon(\mathbf{W}_{s,i}, \mathbf{W}_{c,k})$. To solve it, we adopt the surrogate function obtained by the first-order Taylor expansion. Given the point $\mathbf{W}_{s,i}^n$ in the $n$-th iteration of BCD, the surrogate function $\tilde{\mathbf{W}}_{r,q}^n$ is given by 
\begin{align} \label{eqn:spectral_sca}
  \tilde{\mathbf{W}}_{s,i}^n \triangleq -\| \mathbf{W}_{s,i}^n \|_2 - \mathrm{tr} \big[ \mathbf{u}_{\max,i}^n (\mathbf{u}_{\max,i}^n)^H \left( \mathbf{W}_{s,i} - \mathbf{W}_{s,i}^n \right) \big], 
\end{align}
where $\mathbf{u}_{\max,i}^n$ denotes the eigenvector corresponding to the largest eigenvalue of the matrix $\mathbf{W}_{s,i}^n$. 
As $-\| \mathbf{W}_{s,i} \|_2$ is a concave function of the optimization variable $\mathbf{W}_{s,i}$, it holds that $-\| \mathbf{W}_{s,i} \|_2 \le \tilde{\mathbf{W}}_{s,i}^n$. 
The surrogate function $\tilde{\mathbf{W}}_{c,k}^n$ of $-\|\mathbf{W}_{c,k}\|_2$ is defined similarly to $\tilde{\mathbf{W}}_{r,q}^n$. Thus, the penalty term can be approximated by 
\begin{align}
  \tilde{\varUpsilon}(\mathbf{W}_{s,i}, \mathbf{W}_{c,k}) = &\sum_{i \in \mathcal{M}} \big(\| \mathbf{W}_{s,i} \|_* + \tilde{\mathbf{W}}_{s,i}^n\big) + \sum_{k \in \mathcal{K}} \big(\| \mathbf{W}_{c,k} \|_* + \tilde{\mathbf{W}}_{c,k}^n \big),
\end{align}
which yields the following optimization problem:
\begin{subequations} 
  \begin{align}
    \text{(P3.6):} \max_{\mathbf{\Theta}} \quad &
    f_2(\boldsymbol{\alpha}^{n+1}, \boldsymbol{\beta}^{n+1}, \mathbf{\Theta}) - \frac{1}{\zeta} \tilde{\varUpsilon}(\mathbf{W}_{s,i}, \mathbf{W}_{c,k})\\
    \mathrm{s.t.} \quad & \eqref{c:multi_rate_sca}, \eqref{c:e_constraint} - \eqref{c:semidefinite},
  \end{align}
\end{subequations}
which is a convex QSDP with respect to $\mathbf{\Theta}$ and can be solved by convex solvers such as CVX \cite{cvx}.
It is pointed out in \cite{razaviyayn2013unified} that the BCD method can converge to a stationary point though the SCA is exploited. 

It is worth noting that the value of the regularization parameter $\zeta$ has a significant influence on the obtained solutions. 
Specifically, in order to ensure the nearly rank-one solutions, the value of $\zeta$ should be sufficient small, i.e., $\zeta \rightarrow 0$ and $1/\zeta \rightarrow +\infty$.
However, the objective function is dominated by the penalty term in this case, thereby leading to the insufficiency of maximizing the communication throughput.
Therefore, we firstly set $\zeta$ to a large value for obtaining a good start point with respect to the communication throughput. Then, the value of $\zeta$
is gradually reduced to a sufficiently small value for obtaining a rank-one solution to $\mathbf{W}_{s,i}, \forall i \in \mathcal{M}$ and $\mathbf{W}_{c,k}, \forall k \in \mathcal{K}$.
To this end, we introduce an outer layer nested with the inner BCD layer, where the parameter $\zeta$ is updated following
\begin{equation}
  \zeta = \rho \zeta, 0 < \rho < 1.
\end{equation} 

Based on the approaches proposed above, problem (P3) can be solved in double-layer iterations.
In the inner layer, problem (P3) is transformed to (P3.6) and solved through BCD and SCA with the fixed $\zeta$. In the outer layer, the value of $\zeta$    
is gradually decreased to ensure a rank-one solution. The inner layer is terminated when the communication throughput of the system converges, i.e.,
\begin{equation}
  | (R^N)^n - (R^N)^{n-1}| \le \tau_1,
\end{equation}
where $(R^N)^n$ and $(R^N)^{n-1}$ are the throughputs in the $n$-th and $(n-1)$-th inner iterations, respectively, while the outer layer is terminated when the penalty term is smaller than a pre-defined threshold, i.e.,
\begin{equation} \label{eqn:rank_difference}
  \varUpsilon(\mathbf{W}_{s,i}, \mathbf{W}_{c,k}) \le \tau_2.
\end{equation} 
\begin{algorithm}[htb]
  \caption{Proposed double-layer FP-SCA-based BCD algorithm for solving the problem (P3).}
  \label{alg:A}
  \begin{algorithmic}[1]
        \STATE{Choose a non-zero value of $\epsilon_1$.}
        \STATE{Initialize the feasible $\mathbf{\Theta}^0$ and the parameter $\zeta$.}
        \STATE{Calculate the communication throughput $(R^N)^0$ based on the given $\mathbf{\Theta}^0$.}
        \REPEAT
        \STATE{Set iteration index $n=0$ for inner layer.}
        \REPEAT
            \STATE{Update $\boldsymbol{\alpha}^{n+1}$ by \eqref{eqn:opt_alpha} with $\mathbf{\Theta}^n$. }
            \STATE{Update $\boldsymbol{\beta}^{n+1}$ by \eqref{eqn:opt_beta} with $\boldsymbol{\alpha}^{n+1}$ and $\mathbf{\Theta}^n$. }
            \STATE{Update $\mathbf{\Theta}^{n+1}$ by solving (P3.6) with $\boldsymbol{\alpha}^{n+1}$, $\boldsymbol{\beta}^{n+1}$, and $\mathbf{\Theta}^n$.}
            \STATE{$n = n+1$.}
        \UNTIL{$| (R^N)^{n} - (R^N)^{n-1}| \le \tau_1$.}
        \STATE{$\mathbf{\Theta}^0 = \mathbf{\Theta}^n$, $\zeta = \rho \zeta$.}
        \UNTIL{$\varUpsilon(\mathbf{W}_{s,i}, \mathbf{W}_{c,k}) \le \tau_2$.}
    \end{algorithmic}
\end{algorithm}
\noindent The detailed algorithm for solving the problem (P3) is 
summarized in Algorithm \ref{alg:A}.
The computational burden of this algorithm mainly arise from updating $\boldsymbol{\alpha}$, updating $\boldsymbol{\beta}$ and solving 
the QSDP (P3.6), the computational complexity of which are $\mathcal{O}(KN^2)$, $\mathcal{O}(KN^2)$, and $\mathcal{O}( (K+Q+1)^{6.5} N^{6.5}\log(1/e) )$, respectively. 
The parameter $e$ denotes the solution accuracy for the QSDP using the primal-dual interior-point algorithm.
Thus, for $I_i$ inner iterations and $I_o$ outer iterations, the overall complexity of Algorithm \ref{alg:A} is $\mathcal{O}( I_o I_i ((K+Q+1)^{6.5} N^{6.5}\log(1/e) + 2KN^2) )$.

\subsection{Globally Optimal Solution to (P3) for the Special Single-User Scenario} \label{sec:single}
The single-user scenario is a special case of the general multiple-user scenario with $K=1$. Thus, the related MOOP can also be suboptimally solved using Algorithm \ref{alg:A}. 
However, in this subsection, we obtain the globally optimal solution to the MOOP in the single-user scenario by equivalently transforming it to a convex QSDP. 
Firstly, in the single-user scenario, the communication throughput $R^N$ in \eqref{eqn:throughput_I} is reduced to 
\begin{align} \label{eqn:rate_single}
  R^N =& \sum_{i \in \mathcal{M}} \log_2 \left( 1 +  \frac {\mathbf{h}^H \mathbf{W}_{s,i} \mathbf{h}} { \mathbf{h}^H \mathbf{B}_q \mathbf{h} + \mathbf{h}^H \mathbf{W} \mathbf{h} + \mathbf{h}^H \mathbf{R}_{\tilde{\mathbf{s}}} \mathbf{h} + \sigma_n^2 } \right) + \log_2 \left( 1 + \frac{\mathbf{h}^H \mathbf{W} \mathbf{h}}{ \mathbf{h}^H \mathbf{R}_{\tilde{\mathbf{s}}} \mathbf{h} + \sigma_n^2 } \right) \nonumber\\
  \overset{(a)}{=} &\log_2 \left( 1 + \frac{ \mathbf{h}^H ( \sum_{i \in \mathcal{M}} \mathbf{W}_{s,i} + \mathbf{W}) \mathbf{h} }{ \mathbf{h}^H \mathbf{R}_{\tilde{\mathbf{s}}} \mathbf{h} + \sigma_n^2 } \right) = \log_2 \left( 1 + \frac{\mathbf{h}^H (\mathbf{R}_{\mathbf{x}} - \mathbf{R}_{\tilde{\mathbf{s}}}) \mathbf{h}}{ \mathbf{h}^H \mathbf{R}_{\tilde{\mathbf{s}}} \mathbf{h} + \sigma_n^2 } \right),
\end{align}
where $\mathbf{h} = \mathbf{h}_1$, $\mathbf{B}_i = \sum_{q \in \mathcal{M}, q > i} \mathbf{W}_{s,q}$, and $\mathbf{W} = \mathbf{W}_{c,1}$. The equality $(a)$ is because of the fact that $\log_2(1 + \frac{A}{B+C}) + \log_2(1 + \frac{B}{C}) = \log_2(1 + \frac{A+B}{C})$. It can be observed from \eqref{eqn:rate_single} that if one more beam of the sensing signal is exploited for transmitting information, the value of $\mathbf{h}^H \mathbf{R}_{\tilde{\mathbf{s}}} \mathbf{h}$ will be decreased, i.e., less interference power. Meanwhile, the value of $\mathbf{h}^H (\mathbf{R}_{\mathbf{x}} - \mathbf{R}_{\tilde{\mathbf{s}}}) \mathbf{h}$ will be increased accordingly, i.e., more effective communication power. Therefore, through the proposed NOMA-inspired scheme, we can transform all the interference power into effective communication power to obtain the maximum communication throughput when all the beams of the sensing signal are exploited to transmit information. In this case, the communication throughput $R^N$ becomes $R^N= \log_2 \left( 1 + \frac{\mathbf{h}^H \mathbf{R}_{\mathbf{x}} \mathbf{h}}{\sigma_n^2 } \right)$. The related MOOP is given by 
\begin{subequations}
    \begin{align}
        \text{(P4):} \text{Q1}: \max_{\mathbf{R}_{\mathbf{x}}} \quad & \log_2\left( 1 + \frac{\mathbf{h}^H \mathbf{R}_{\mathbf{x}} \mathbf{h}}{\sigma_n^2 } \right), \\
        \text{Q2}: \min_{\delta, \mathbf{R}_{\mathbf{x}} } \quad & \frac{1}{L} \sum_{l=1}^{L}\left| \delta \phi(\theta_l) - \mathbf{a}^H(\theta_l) \mathbf{R}_{{\bf x}} \mathbf{a}(\theta_l)  \right|^2 \\
        \label{c:transmit_power_single}
        \mathrm{s.t.} \quad & \mathrm{tr}(\mathbf{R}_{\mathbf{x}}) = P_t, \\
        \label{c:semidefinite_single}
        & \mathbf{R}_{\mathbf{x}} \succeq 0.
    \end{align}
\end{subequations}
Then, by exploiting the $\epsilon$-constraint method, the resultant SOOP for a specific value of $\epsilon_2$ is given by 
\begin{subequations}
    \begin{align}
        \text{(P4.1):} \min_{\delta, \mathbf{R}_{\mathbf{x}} } \quad & \frac{1}{L} \sum_{l=1}^{L}\left| \delta \phi(\theta_l) - \mathbf{a}^H(\theta_l) \mathbf{R}_{{\bf x}} \mathbf{a}(\theta_l)  \right|^2 \\
        \mathrm{s.t.} \quad & \mathbf{h}^H \mathbf{R}_{\mathbf{x}} \mathbf{h} - (2^{\epsilon_2}-1)  \sigma_n^2 \ge 0,\\
        & \eqref{c:transmit_power_single}, \eqref{c:semidefinite_single}.
    \end{align}
\end{subequations}
It is clear that problem (P4.1) is a convex QSDP and can be optimally solved by the convex solvers like CVX \cite{cvx}. When the primal-dual interior-point algorithm
is exploited, (P4.1) can be solved with the worst complexity of $\mathcal{O}(N^{6.5} \log(1/e))$, where $e$ denotes the solution accuracy. 

\subsection{Performance Comparison with the State-of-the-Art ISAC Frameworks}
In this subsection, we compare the proposed NOMA-inspired ISAC framework with the state-of-the-art ISAC frameworks, based on which the property 
of the proposed framework is analyzed.

\subsubsection{State-of-the-Art ISAC Frameworks}
In the state-of-the-art ISAC frameworks, the sensing signal is regarded as a harmful interference to communication.
Thus, the received signal at user $k$ is given by 
\begin{align} 
  y_k(t) = \mathbf{h}_k^H \mathbf{w}_{c,k} c_k(t) + \underbrace{\mathbf{h}_k^H \sum_{j \in \mathcal{K} \setminus k} \mathbf{w}_{c,j} c_j(t)}_{\text{inter-user interference}} + \underbrace{\mathbf{h}_k^H \mathbf{s}}_{\text{sensing signal}} + n_k.
\end{align} 
The achievable rate of the desired unicast signal at user $k$ is given by 
\begin{equation} \label{eqn:rate_exsiting}
    R_k^{C} = \log_2 \left( 1 + \frac{|\mathbf{h}_k^H \mathbf{w}_{c,k}|^2}{ \sum_{j \in \mathcal{K} \setminus k } |\mathbf{h}_k^H \mathbf{w}_{c,j}|^2 + p  \mathbf{h}_k^H \mathbf{R}_{\mathbf{s}} \mathbf{h}_k + \sigma_n^2} \right),
\end{equation} 
where the parameter $p$ is the indicator of the following two state-of-the-art ISAC frameworks:
\begin{itemize}
  \item \textbf{No sensing interference cancellation (No-SenIC) ISAC} \cite{liu2020beamforming,liu2021cram}: The communication receiver has no capability to mitigate the sensing interference. Thus, the desired unicast signal is decoded subject to interference from the sensing signal.
  In this case, $p=1$. 

  \item \textbf{Ideal sensing interference cancellation (Ideal-SenIC) ISAC} \cite{hua2021optimal}: The communication receiver is assumed to be capable of ideally removing the sensing interference from $y_k$. 
  In this case, $p=0$. 

\end{itemize}

\subsubsection{General Multiple-user Scenario}
Following the same procedure in Section \ref{sec:multiple}, in the multiple-user scenario, the SOOP for the Ideal-SenIC ISAC and No-SenIC ISAC obtained by the $\epsilon$-constraint method is given by  
\begin{subequations}
    \begin{align}
      \text{(P5.1):} \max_{\delta, \mathbf{W}_c, \mathbf{R}_{\mathbf{s}}} \quad & R^{C} = \sum_{k \in \mathcal{K}} R_k^{C} \\
      \mathrm{s.t.} \quad & \frac{1}{L} \sum_{l=1}^{L}\left| \delta \phi(\theta_l) - \mathbf{a}^H(\theta_l) \mathbf{R}_{{\bf x}} \mathbf{a}(\theta_l)  \right|^2 \le \epsilon_1, \\
      & \mathrm{tr}(\mathbf{W}_c \mathbf{W}_c^H + \mathbf{R}_{\mathbf{s}}) = P_t, \\
      & \mathbf{R}_{\mathbf{s}} \succeq 0.
    \end{align}
\end{subequations}
Despite the non-convexity of the problem (P5.1), it can be solved by the proposed \textbf{Algorithm \ref{alg:A}} by removing the extra information streams and considering the interference term in \eqref{eqn:rate_exsiting}. In the following, we compare the NOMA-inspired ISAC frameworks with the Ideal-SenIC and No-SenIC ISAC by analyzing the solution to (P3.1) and (P5.1).

\begin{proposition} \label{proposition_2}
  \emph{Given a specific value of $\epsilon_1$, let $\bar{R}^{N}$, $\bar{R}^{C0}$, and $\bar{R}^{C1}$
  denote the optimal values of problem (P3.1) with $p=0$, and (P5.1) with $p=1$, respectively. Then, the following inequality holds when $M = N$: 
  \begin{equation} \label{eqn:multiple_rate_ieqn}
    \bar{R}^{N} \ge \bar{R}^{C0} \ge \bar{R}^{C1}.
  \end{equation}}
\end{proposition}
\begin{proof}
  See Appendix A.
\end{proof}

\begin{remark} \label{remark:multiple}
  \emph{\textbf{Proposition \ref{proposition_2}} shows that when the same sensing beampattern accuracy is achieved, the proposed NOMA-inspired ISAC framework can achieve a higher communication throughput compared to the existing Ideal-SenIC ISAC and No-SenIC ISAC frameworks,
  which indicates the benefits of further exploiting the sensing signal for communication. Furthermore, although the inequalities \eqref{eqn:multiple_rate_ieqn}
  is proved under a strong condition $M=N$, our numerical results provided in the next section show that $M=1$ is sufficient. }
\end{remark}

\subsubsection{Special Single-user Scenario}
For the special single-user scenario, the corresponding SOOP of Ideal-SenIC and No-SenIC ISAC can be obtained following the same procedure in Section \ref{sec:single}, which is given by 
\begin{subequations}
  \begin{align}
      \text{(P6.1):} &\min_{\delta, \mathbf{W}, \mathbf{R}_{\mathbf{s}}} \quad \frac{1}{L} \sum_{l=1}^{L}\left| \delta \phi(\theta_l) - \mathbf{a}^H(\theta_l) \mathbf{R}_{{\bf x}} \mathbf{a}(\theta_l)  \right|^2 \\
      \mathrm{s.t.} \quad & \mathbf{h}^H \mathbf{W} \mathbf{h} - (2^{\epsilon_2}-1)(p \mathbf{h}^H \mathbf{R}_{\mathbf{s}} \mathbf{h} + \sigma_n^2 ) \ge 0,\\
      & \mathrm{tr}(\mathbf{W} + \mathbf{R}_{\mathbf{s}}) = P_t, \\
      & \mathbf{W} \succeq 0, \mathbf{R}_{\mathbf{s}} \succeq 0,
  \end{align}
\end{subequations}
Problem (P6.1) is convex for both $p=0$ and $p=1$ and thus can be optimally solved. Although the rank-one 
constraint of $\mathbf{W}$ is omitted in problem (P6.1), there always exists a globally optimal rank-one solution according to \cite[Proposition 1]{hua2021optimal} and \cite[Proposition 2]{hua2021optimal}.

\begin{proposition} \label{proposition_1}
  \emph{Given a specific value of $\epsilon_2$, let $\bar{\mathbf{R}}_{\mathbf{x}}^N$, $\bar{\mathbf{R}}_{\mathbf{x}}^{C0}$, and $\bar{\mathbf{R}}_{\mathbf{x}}^{C1}$   
  denote the globally optimal transmit covariance matrix to (P4.1), (P6.1) with $p=0$, and (P6.1) with $p=1$, respectively. Then, the following equality holds:
  \begin{equation}
      L(\delta^N, \bar{\mathbf{R}}_{\mathbf{x}}^N) = L(\delta^{C0}, \bar{\mathbf{R}}_{\mathbf{x}}^{C0}) = L(\delta^{C1}, \bar{\mathbf{R}}_{\mathbf{x}}^{C1}),
  \end{equation}
  where $\delta^N$, $\delta^{C0}$, and $\delta^{C1}$ denote the corresponding optimal scaling vectors. 
  }
\end{proposition}
\begin{proof}
  See Appendix B.
\end{proof}

\begin{remark} \label{remark_single}
  \emph{\textbf{Proposition \ref{proposition_1}} indicates that the proposed NOMA-inspired ISAC achieves the same optimal performance as the existing SenIC ISAC frameworks in the single-user scenario, which provides the following insights. On the one hand, there is no need of exploiting the sensing signal to transmit information in the single-user scenario. On the other hand, when only one communication user is served, the ISAC system can always achieve optimal performance regardless of the capability of the sensing interference cancellation at the receiver. Thus, sensing interference coordination is not needed.}
\end{remark}
\section{Numerical Results} \label{sec:results}

In this section, the numerical results are provided for characterizing the proposed NOMA-inspired ISAC framework. We assume a dual-functional BS equipped with a ULA with half-wavelength spacing, which works in the tracking mode to sensing $3$ targets in the directions $\Phi = \{-60^\circ, 0^\circ, 60^\circ\}$. The transmit power budget at BS and the noise power at the communication users are set to $P_t = 20$ dBm and $\sigma_n^2 = -80$ dBm, respectively. The channels between BS and communication users are assumed to experience Rayleigh fading and $80$ dB path loss. Given the directions of sensing targets, the desired beampattern is given by 
\begin{equation}
    \phi (\theta_l) = \begin{cases}
        1, &\theta_l \in [\varphi - \frac{\Delta}{2}, \varphi - \frac{\Delta}{2}], \forall \varphi \in \Phi, \\
        0, &\mathrm{otherwise},
    \end{cases}
\end{equation}
where $\Delta$ is the desired beam width, which is set to $10^\circ$ in the simulation, and the angle grid $\{\theta_l\}_{l=1}^L$ is set to $[-\frac{\pi}{2}: \frac{\pi}{100} : \frac{\pi}{2}]$. In order to show the sensing performance directly, we access it depending on the following matching error:
\begin{equation}
    \epsilon = \frac{1}{L} \sum_{l=1}^{L}\left| P(\theta, \mathbf{R}_\mathbf{x}^S) - P(\theta, \mathbf{R}_\mathbf{x})  \right|^2.
\end{equation}
Here, $\mathbf{R}_\mathbf{x}^S$ denotes the optimal covariance matrix for the sensing-only system, which can be obtained by solving the following convex optimization problem:
\begin{subequations}
    \begin{align}
        \min_{\delta, \mathbf{R}_{\mathbf{x}} \succeq 0. } \quad & L(\delta, \mathbf{R}_\mathbf{x}) = \frac{1}{L} \sum_{l=1}^{L}\left| \delta \phi(\theta_l) - \mathbf{a}^H(\theta_l) \mathbf{R}_{{\bf x}} \mathbf{a}(\theta_l)  \right|^2 \\
        \mathrm{s.t.} \quad & \mathrm{tr}(\mathbf{R}_{\mathbf{x}}) = P_t.
    \end{align}
\end{subequations}

\subsection{Multiple User Scenario}
In this section, numerical results of the multiple-user scenario are presented. We set $N=8$ and $K=5$. For Algorithm \ref{alg:A}, the initial penalty factor is set to $\zeta = 10^2$ and its reduction factor is set to $\rho = 0.2$. The convergence thresholds of the inner and outer iterations are set to $\tau_1 = 10^{-2}$ and $\tau_2 = 10^{-4}$, respectively.

\subsubsection{Benchmark schemes}
Apart from Ideal-SenIC ISAC and No-SenIC ISAC, we also consider the following benchmark schemes for comparison.
\begin{itemize}
    \item \textbf{Communication signal only (Com) ISAC} \cite{liu2018mu,liu2018toward}: In this scheme, the BS only transmits the communication signals for achieving both communication and sensing. Thus, the covariance matrix of the transmit signal is $\mathbf{R}_\mathbf{x} = \sum_{j \in \mathcal{K}} \mathbf{w}_{c,j} \mathbf{w}_{c,j}^H$. 
    The achievable rate at user $k$ is given by
    \begin{equation}
        R_k^{Com} = \log_2 \left( 1 + \frac{| \mathbf{h}_k^H \mathbf{w}_{c,k}|^2}{ \sum_{j \in \mathcal{K} \setminus k } |\mathbf{h}_k^H \mathbf{w}_{c,j}|^2 + \sigma_n^2} \right).
    \end{equation}
    The communication throughput is given by $R^{Com} = \sum_{k \in \mathcal{K}} R_k^{Com}$. The resultant MOOP in the multiple user scenario can be solved by the proposed Algorithm \ref{alg:A}.

    \item \textbf{NOMA-inspired sensing interference cancellation (NOMA-SenIC) ISAC} \cite{wang2021noma}: This scheme is essentially the proposed framework without extra information transmission. In this scheme, the information embedded into the sensing signal is merely used to carry out SIC. In this case, the communication throughput is $R^{N0} = \sum_{k \in \mathcal{K}} R_k^{N}$. The resultant MOOP can also be solved by the proposed Algorithm \ref{alg:A}.
\end{itemize}

\begin{figure}[t!]
    \centering
    \begin{minipage}[t]{0.4\textwidth}
        \vspace{0pt}
        \centering
        \includegraphics[width=1\textwidth]{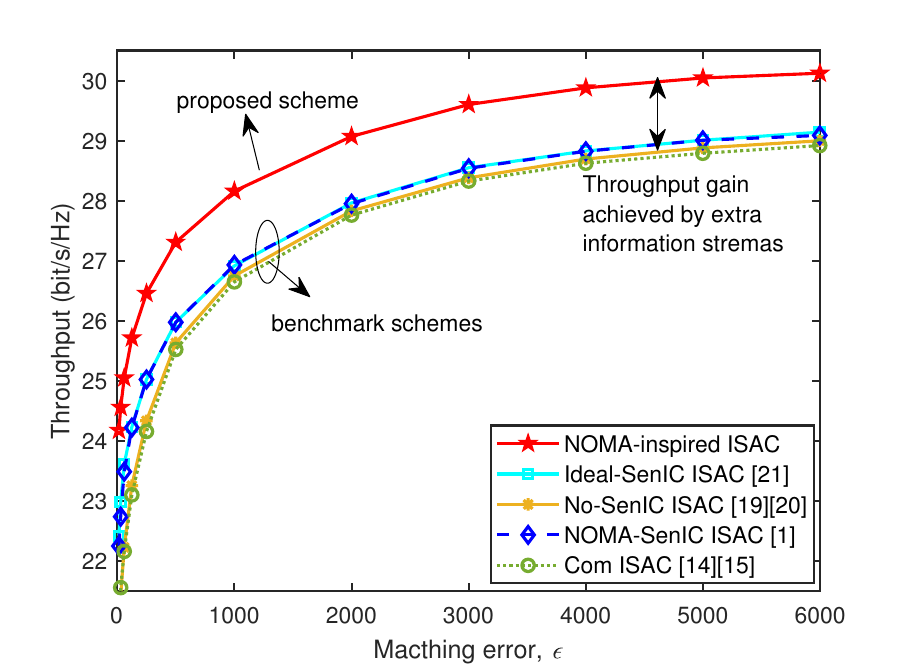}
        \caption{Trade-off region in the multiple-user scenario with $M=1$. }
        \label{fig:multiple_user_pareto}
    \end{minipage}\hspace{3mm}
    \begin{minipage}[t]{0.4\textwidth}
        \vspace{0pt}
        \centering
        \includegraphics[width=1\textwidth]{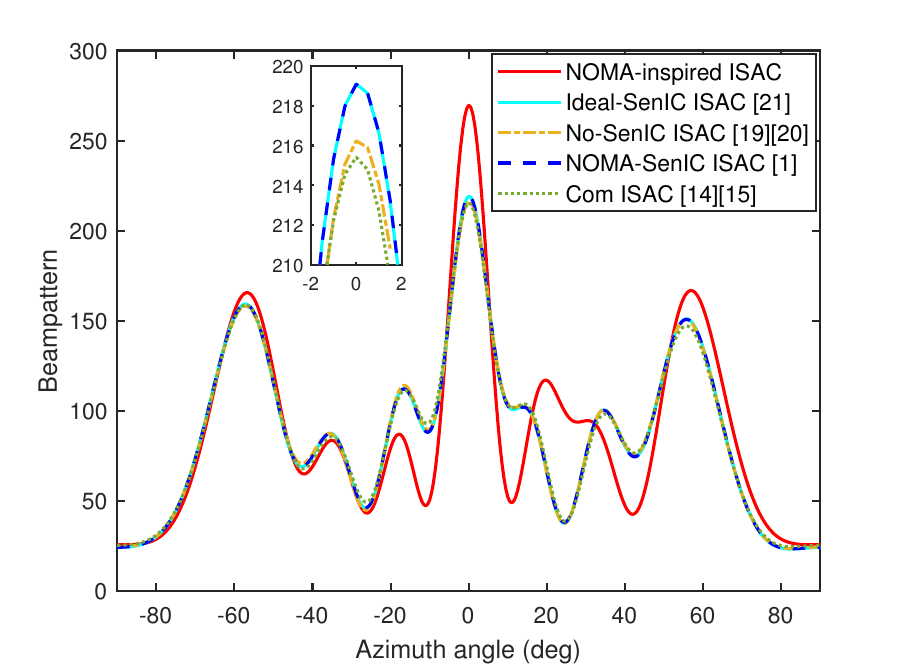}
        \caption{Obtained sensing beampattern in the multiple-user scenario with throughput $R=26.4$ bit/s/Hz and $M=1$.}
        \label{fig:multiple_user_beampattern}
    \end{minipage}
    \vspace{-0.5cm}
\end{figure}

\subsubsection{Trade-off Region}
In Fig. \ref{fig:multiple_user_pareto}, the trade-off region between the communication throughput and the beampattern matching error is studied, which is obtained by averaging over $50$ random channel realizations for each scheme. For the proposed NOMA-inspired ISAC framework, we set $M=1$. It can be seen that the proposed NOMA-inspired ISAC framework achieves a larger trade-off region than the other benchmark schemes. In particular, although only one beam of the sensing signal is exploited to transmit the extra information stream, the NOMA-inspired ISAC still significantly outperforms the Ideal-SenIC ISAC. The reason behind this can be explained as follows. The obtained sensing signal is dominated by the beam corresponding to the largest eigenvalue, which can be used to carry the extra information stream for achieving a considerable rate. Meanwhile, the other beams of the sensing signal have very low power, thus making negligible interference to communication. This fact is also revealed by the result achieved by the NOMA-SenIC ISAC framework, where even only one beam of the sensing signal is eliminated but the achieved performance is close to that of Ideal-SenIC ISAC. Finally, the No-SenIC ISAC and Com-ISAC frameworks have the worst performance due to the sensing interference and the lack of transmitting DoFs, respectively. These results are consistent with \textbf{Proposition \ref{proposition_2}} and \textbf{Remark \ref{remark:multiple}}.

\subsubsection{Sensing Beampattern}
In Fig. \ref{fig:multiple_user_beampattern}, the sensing beampatterns obtained via different schemes over one channel realization when $R=26.4$ bit/s/Hz and $M=1$ are depicted. One can observe that the proposed NOMA-inspired ISAC framework is capable of achieving the best sensing beampattern, where there is a noticeable gain of power in the target directions and less power leakage in the undesired directions. When there is no extra information stream, i.e., in NOMA-SenIC ISAC, the beampattern is still close to that implemented by the Ideal-SenIC ISAC framework. Additionally, the No-SenIC ISAC and Com-SenIC ISAC frameworks have the largest beampattern distortion.

\begin{figure}[t!]
    \centering
    \begin{minipage}[t]{0.4\textwidth}
        \vspace{0pt}
        \centering
        \includegraphics[width=1\textwidth]{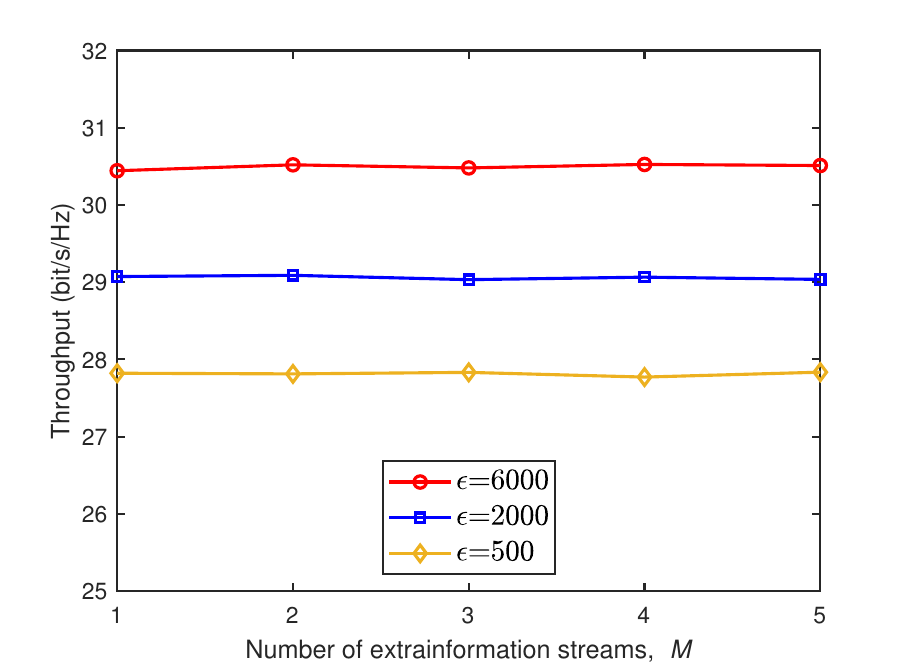}
        \caption{Communication throughput versus the number of extra information streams, $M$.}
        \label{fig:rate_vs_M}
    \end{minipage}\hspace{3mm}
    \begin{minipage}[t]{0.4\textwidth}
        \vspace{0pt}
        \centering
        \includegraphics[width=1\textwidth]{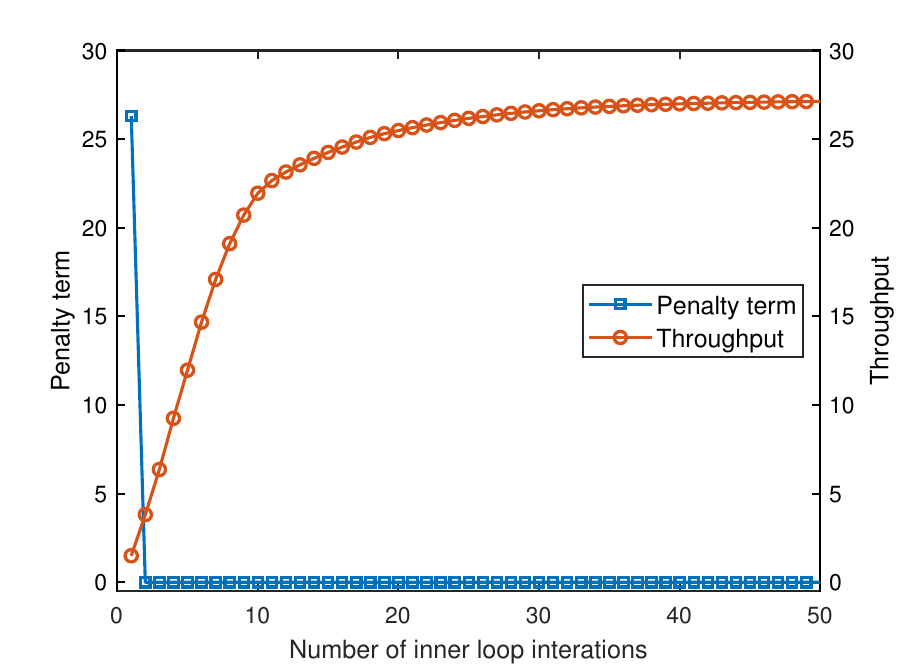}
        \caption{Convergence of Algorithm \ref{alg:A} when $M=1$ and $\epsilon=2 \times 10^3$.}
        \label{fig:convergence}
    \end{minipage}
    \vspace{-0.3cm}
\end{figure}

\subsubsection{Throughput versus $M$} 
In Fig. \ref{fig:rate_vs_M}, the relationship between the communication throughput and the number of the information-bearing sensing beams $M$ is investigated. The results are obtained by averaging over $50$ random channel realizations. It is clear that when the value of $M$ increases, namely more beams of the sensing signal are exploited to transmit the extra information streams and fewer beams interference with communication, the communication throughput is almost unchanged. This further verifies that the power of the sensing signal is mainly concentrated in the beam corresponding to the largest eigenvalue and other beams have negligible interference in communication. Thus, the simplest design of the proposed NOMA-inspired ISAC framework with $M=1$ is sufficient in practice. 

\subsubsection{Convergence of Algorithm \ref{alg:A}}
In Fig. \ref{fig:convergence}, we demonstrate the convergence behavior of the proposed algorithm over one random channel realization when $M=1$ and $\epsilon=2 \times 10^3$. We can observe that the throughput gradually increases to a stable value, while the penalty value quickly converges to almost zero as the number of inner iterations increases, which indicates that the resultant beamformers are rank-one. 

\subsection{Single-User Scenario}
In this section, the numerical results of the single-user scenario are presented, where the proposed NOMA-inspired ISAC is compared with the existing Ideal-SenIC ISAC and No-SenIC ISAC.
\begin{figure}[t!]
    \centering
    \begin{minipage}[t]{0.4\textwidth}
        \vspace{0pt}
        \centering
        \includegraphics[width=1\textwidth]{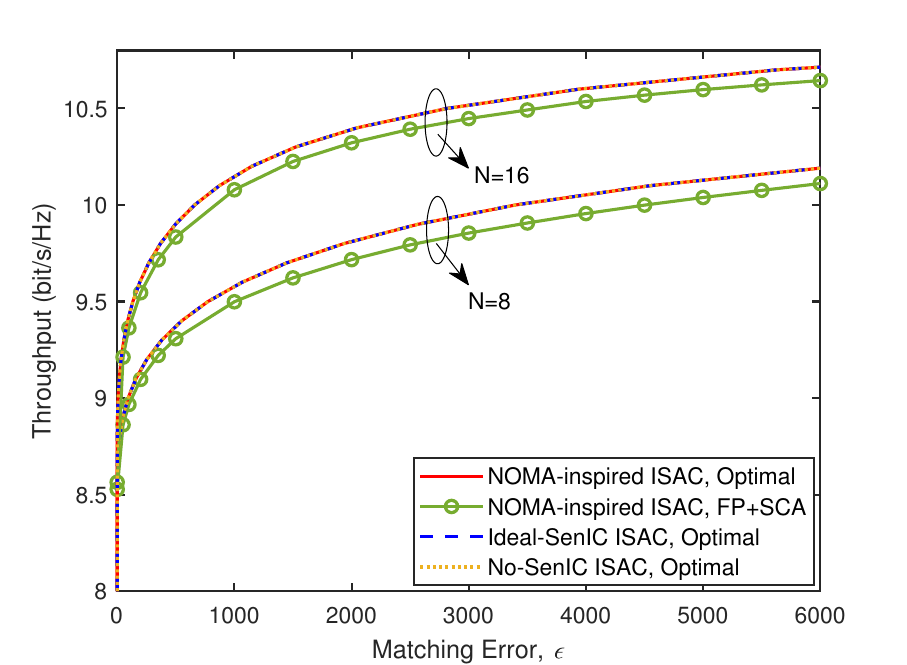}
        \caption{Trade-off region in the single user scenario.}
        \label{fig:single_user_pareto}
    \end{minipage}\hspace{3mm}
    \begin{minipage}[t]{0.4\textwidth}
        \vspace{0pt}
        \centering
        \includegraphics[width=1\textwidth]{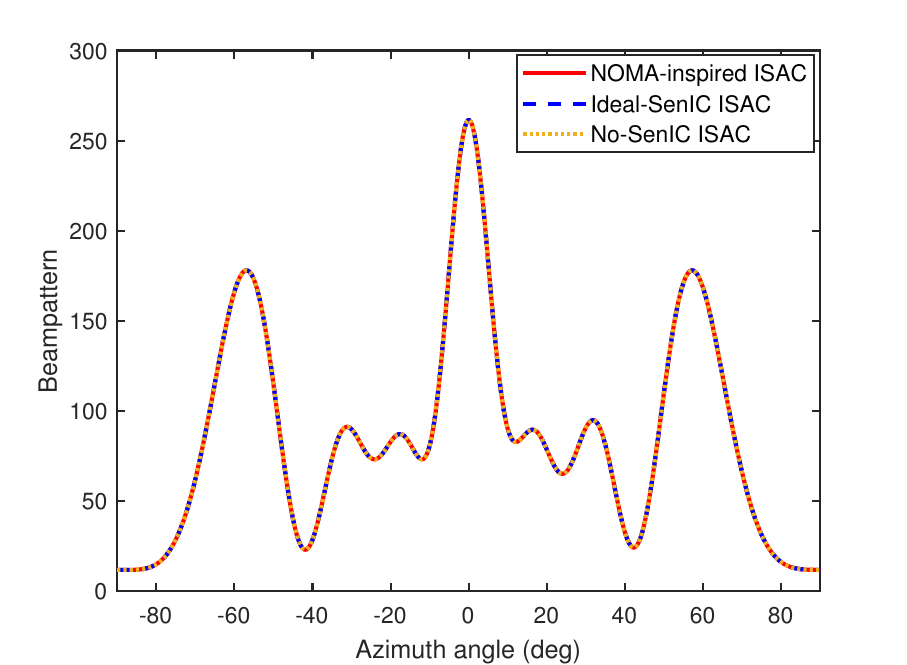}
        \caption{Obtained transmit beampattern in the single user scenario with $N=8$ and throughput $R=9.0$ bit/s/Hz}
        \label{fig:single_user_beampattern}
    \end{minipage}
\end{figure}

\subsubsection{Trade-off Region}
In Fig. \ref{fig:single_user_pareto}, we investigate the trade-off region of the communication throughput and the beampattern matching error, which is obtained by averaging over one random channel realization. It can be observed that compared to the suboptimal solution obtained via Algorithm \ref{alg:A}, the globally optimal solution via QSDP achieves a better trade-off region. Furthermore, one can see that the proposed NOMA-inspired ISAC achieve exactly the same trade-off region as the existing Ideal-SenIC ISAC and No-SenIC ISAC, which is consistent with \textbf{Proposition \ref{proposition_1}} and \textbf{Remark \ref{remark_single}}.

\subsubsection{Sensing Beampattern}
In Fig. \ref{fig:single_user_beampattern}, we study the obtained sensing beampattern over one random channel realization by NOMA-inspired ISAC, Ideal-SenIC ISAC, and the No-SenIC ISAC when the communication throughput is $R=9.0$ bit/s/Hz. We set $N=8$. It can be observed that all three frameworks can achieve dominant peaks in the target directions. Furthermore, the transmit beampatterns achieved by them are exactly the same, which is also consistent with \textbf{Proposition \ref{proposition_1}} and \textbf{Remark \ref{remark_single}}.

\section{Conclusions} \label{sec:conclusion}
In this paper, a NOMA-inspired ISAC framework was proposed, where the transmit sensing signal at BS was further exploited for conveying the extra information streams based on the concept of NOMA. A tailor-made MOOP designing the transmit beamforming was formulated with the aim of maximizing the communication throughput and minimizing the sensing beampattern matching error. For the multiple-user scenario, the formulated MOOP was transformed to a SOOP via $\epsilon$-constrtaint method and solved by the proposed double-layer FP-SCA-based BCD algorithm. For the special-single user scenario, the globally optimal solution was obtained by transforming the MOOP  to a convex QSDP. It was shown by the theoretical and numerical results that the trade-off region and sensing beampattern achieved by the proposed NOMA-inspired ISAC in the multiple-user scenario are better than those achieved by the existing SenIC ISAC frameworks while becoming the same in the single-user scenario. Finally, it was revealed that using only one beam of the sensing signal to convey the extra information stream still has very high efficiency.

\section*{Appendix~A: Proof of Proposition \ref{proposition_2}} \label{appendix_2}
\renewcommand{\theequation}{A.\arabic{equation}}
\setcounter{equation}{0}
Firstly, let $(\breve{\delta}, \{\breve{\mathbf{w}}_{c,k}\}, \breve{\mathbf{R}}_{\mathbf{s}})$ denote any feasible solutions to (P5.1). 
Then, it holds that 
\begin{align}
    &\breve{R}_k^{C0} = \log_2 \left( 1 + \frac{| \mathbf{h}_k^H \breve{\mathbf{w}}_{c,k}|^2}{ \sum_{j \in \mathcal{K} \setminus k} |\mathbf{h}_k^H \breve{\mathbf{w}}_{c,j}|^2 + \sigma_n^2} \right) \nonumber\\
    & \ge \log_2 \left( 1 + \frac{| \mathbf{h}_k^H \breve{\mathbf{w}}_{c,k}|^2}{ \sum_{j \in \mathcal{K} \setminus k} |\mathbf{h}_k^H \breve{\mathbf{w}}_{c,j}|^2 +  \mathbf{h}_k^H \breve{\mathbf{R}}_{\mathbf{s}} \mathbf{h}_k + \sigma_n^2} \right) = \breve{R}_k^{C1},
\end{align}  
which is because the matrix $\breve{\mathbf{R}}_\mathbf{s}$ is positive semidefinite. Thus, we have
\begin{equation}
    \breve{R}^{C0} = \sum_{k \in \mathcal{K}} \breve{R}_k^{C0} \ge \sum_{k \in \mathcal{K}} \breve{R}_k^{C1} = \breve{R}^{C1}.
\end{equation} 
Then, when $M=N$, we define the $\breve{\mathbf{w}}_{s,i}, \forall i \in \mathcal{M}$ as follows: 
\begin{equation}
    \breve{\mathbf{w}}_{s,i} = \begin{cases}
        \sqrt{\breve{\lambda}_i} \breve{\mathbf{v}}_i, &0 \le i \le \mathrm{rank}(\breve{\mathbf{R}}_{\mathbf{s}}) \\
        \mathbf{0}_{N \times 1}, &\mathrm{rank}(\breve{\mathbf{R}}_{\mathbf{s}}) < i \le N
    \end{cases},
\end{equation}
where $\breve{\lambda}_q$ and $\breve{\mathbf{v}}_q$ are the $q$-th eigenvalue and eigenvector of $\breve{\mathbf{R}}_{\mathbf{s}}$, respectively. Then, it can be verified that 
$(\{ \breve{\mathbf{w}}_{c,k} \}, \{\breve{\mathbf{w}}_{s,i}\}, \breve{\mathbf{R}}_{\tilde{\mathbf{s}}}=\mathbf{0}_{N \times N} )$ is feasible solutions to (P3.1).
As a result, the achievable unicast rate in (P3.1) is given by
\begin{equation}
    \breve{R}_k^{N} = \log_2 \left( 1 + \frac{| \mathbf{h}_k^H \breve{\mathbf{w}}_{c,k}|^2}{ \sum_{j \in \mathcal{K} \setminus k} |\mathbf{h}_k^H \breve{\mathbf{w}}_{c,j}|^2 + \sigma_n^2} \right) = \breve{R}_k^{C0}.
\end{equation}
The extra rate $\hat{R}_i^{N} \ge 0, \forall i \in \mathcal{M}$ achieved by the sensing signal in (P3.1) can be calculated according to \eqref{eqn:multicast_rate_1} and \eqref{eqn:multicast_rate_2}. Thus, it holds that 
\begin{equation}
    \breve{R}^{N} = \sum_{i \in \mathcal{M}} \hat{R}_i^{N} + \sum_{k \in \mathcal{K}} \breve{R}_k^{N} \ge \sum_{k \in \mathcal{K}} \breve{R}_k^{C0} = \breve{R}^{C0}.
\end{equation} 
Based on above analysis, we can conclude that for any feasible solutions $(\{ \breve{\mathbf{w}}_{c,k} \}, \breve{\mathbf{R}}_{\mathbf{s}})$ to (P5.1), it always holds that $\breve{R}^{C0} \ge \breve{R}^{C1}$. Furthermore, there always exists feasible solutions to (P3.1) such that $\breve{R}^{N} \ge \breve{R}^{C0}$. Therefore, for the optimal values $\bar{R}^{N}$, $\bar{R}^{C0}$, and $\bar{R}^{C1}$ of the three problems, it must hold that $\bar{R}^{N} \ge \bar{R}^{C0} \ge \bar{R}^{C1}$ and \textbf{Proposition \ref{proposition_2}} is finally proved.

\section*{Appendix~B: Proof of Proposition \ref{proposition_1}} \label{appendix_1}
\renewcommand{\theequation}{B.\arabic{equation}}
\setcounter{equation}{0}
In order to prove Proposition \ref{proposition_1}, we firstly give the following three lemmas.
\begin{lemma} \label{lemma_3}
    Given a specific value of $\epsilon_2$, it always holds that $L(\delta^{C1}, \bar{\mathbf{R}}_{\mathbf{x}}^{C1}) \ge L(\delta^{C0}, \bar{\mathbf{R}}_{\mathbf{x}}^{C0})$. 
\end{lemma}
\begin{proof}
    It can be observed that any feasible solutions to (P6.1) with $p=1$ are also feasible to (P6.1) with $p=0$ but not vice versa, which means the feasible region of (P6.1) with $p=0$ is larger than that of (P6.1) with $p=1$. Thus, the optimal beampattern matching error achieved by (P6.1) with $p=0$ is always smaller or equal to that achieved by (P6.1) with $p=1$ for the same value of $\epsilon_2$. As a result, \textbf{Lemma \ref{lemma_3}} is proved.
\end{proof}

\begin{lemma} \label{lemma_4}
    Given a specific value of $\epsilon_2$, it always holds that $L(\delta^{C0}, \bar{\mathbf{R}}_{\mathbf{x}}^{C0}) \ge L(\delta^N, \bar{\mathbf{R}}_{\mathbf{x}}^{N})$. 
\end{lemma}
\begin{proof}
    In (P6.1) with $p=0$, only the communication signal is exploited for communication, while in (P4.1), both communication and sensing signals are exploited for communication. Thus, (P4.1) can be rewritten as 
    \begin{subequations}
        \begin{align}
            \text{(P7.1):} \min_{\delta, \mathbf{W}, \mathbf{R}_\mathbf{s} } \quad & L(\delta, \mathbf{R}_\mathbf{x}) \\
            \mathrm{s.t.} \quad & \mathbf{h}^H (\mathbf{W} + \mathbf{R}_\mathbf{s}) \mathbf{h} - (2^{\epsilon_2}-1)\sigma_n^2  \ge 0, \\
            & \mathrm{tr}(\mathbf{W} + \mathbf{R}_\mathbf{s}) = P_t, \\
            & \mathbf{W} \succeq 0, \mathbf{R}_\mathbf{s} \succeq 0.
        \end{align}
    \end{subequations}

    \noindent It can be observed that any feasible solutions to (P6.1) with $p=0$ are also feasible to (P7.1) but not vice versa. Thus, \textbf{Lemma \ref{lemma_4}} is proved.
\end{proof}

\begin{lemma} \label{lemma_5}
    Given a specific value of $\epsilon_2$, it always holds that $L(\delta^{N}, \bar{\mathbf{R}}_{\mathbf{x}}^{N}) = L(\delta^{C1}, \bar{\mathbf{R}}_{\mathbf{x}}^{C1})$. 
\end{lemma}
\begin{proof}
    Given the global optimum $\{\delta^\star, \mathbf{W}^\star, \mathbf{R}_{\mathbf{s}}^\star \}$ to (P6.1) with $p=1$, we now show that
    $\mathbf{R}_\mathbf{x}^\star = \mathbf{W}^\star + \mathbf{R}_{\mathbf{s}}^\star $ is also the global optimum to (P4.1).
    Firstly, it holds that 
    \begin{align}
        \mathbf{h}^H \mathbf{R}_\mathbf{x}^\star \mathbf{h} = \mathbf{h}^H (\mathbf{W}^\star + \mathbf{R}_{\mathbf{s}}^\star) \mathbf{h} 
        \overset{(a)}{\ge} \mathbf{h}^H \mathbf{W}^\star \mathbf{h} 
        \overset{(b)}{\ge} (2^{\epsilon_2}-1)( \mathbf{h}^H \mathbf{R}_{\mathbf{s}}^\star \mathbf{h} + \sigma_n^2 )
        \overset{(c)}{\ge} (2^{\epsilon_2}-1)\sigma_n^2,
    \end{align}

    \noindent where $(a)$ and $(c)$ are because $\mathbf{R}_{\mathbf{s}}^\star$ is positive semidefinite and $(b)$ is because of the feasibility of $\{ \mathbf{W}^\star, \mathbf{R}_{\mathbf{s}}^\star \}$ to (P6.1) with $p=1$.
    It also holds that $\mathrm{tr}(\mathbf{R}_\mathbf{x}^\star) = \mathrm{tr}(\mathbf{W}^\star + \mathbf{R}_{\mathbf{s}}^\star) = P_t$. Therefore, $\mathbf{R}_\mathbf{x}^\star$ is a feasible solution to (P4.1). 
    
    Furthermore, it can be observed for any feasible solution $\mathbf{R}_\mathbf{x}$ to (P4.1), 
    there always exists a solution $\{ \mathbf{W} = \mathbf{R}_\mathbf{x}, \mathbf{R}_\mathbf{s} = 0 \}$ that is feasible to (P6.1) with $p=1$ such that $L(\delta, \mathbf{W} + \mathbf{R}_\mathbf{s}) = L(\delta, \mathbf{R}_\mathbf{x})$. 
    Therefore, it can be shown that 
    \begin{equation}
        L(\delta^\star, \mathbf{R}_\mathbf{x}^\star) = L(\delta^\star, \mathbf{W}^\star + \mathbf{R}_{\mathbf{s}}^\star) \le L(\delta, \mathbf{W} + \mathbf{R}_\mathbf{s}) = L(\delta, \mathbf{R}_\mathbf{x}),
    \end{equation} 

    \noindent which indicates that $\mathbf{R}_\mathbf{x}^\star = \mathbf{W}^\star + \mathbf{R}_{\mathbf{s}}^\star $ is a global optimum to (P4.1). \textbf{Lemma \ref{lemma_5}} is proved.
\end{proof}

According to \textbf{Lemma \ref{lemma_3}}, \textbf{Lemma \ref{lemma_4}}, and \textbf{Lemma \ref{lemma_5}}, it holds that $L(\delta^{C1}, \bar{\mathbf{R}}_{\mathbf{x}}^{C1}) \ge L(\delta^{C0}, \bar{\mathbf{R}}_{\mathbf{x}}^{C0}) \ge L(\delta^{N}, \bar{\mathbf{R}}_{\mathbf{x}}^{N}) = L(\delta^{C1}, \bar{\mathbf{R}}_{\mathbf{x}}^{C1})$. Thus, the equality among them must hold and \textbf{Proposition \ref{proposition_1}} is finally proved.

\bibliographystyle{IEEEtran}
\begin{spacing}{1.0}
\bibliography{reference/mybib}
\end{spacing}

\end{document}